\documentclass[sigconf]{acmart}

\AtBeginDocument{%
  \providecommand\BibTeX{{%
    \normalfont B\kern-0.5em{\scshape i\kern-0.25em b}\kern-0.8em\TeX}}}

\usepackage{tikz}
\usepackage[ruled]{algorithm2e} 

\usepackage[english]{babel}
\usepackage{listings}
\usepackage{graphicx,mathtools,amsthm,amsmath}
\usepackage{stmaryrd}
\usepackage{ragged2e}
\usepackage{enumerate}
\usepackage{longtable}
\usepackage{array}
\usepackage{float}
\restylefloat{table}
\usepackage{hyperref}
\usepackage{wrapfig}
\usepackage{rotating}
\usepackage{pifont}
\usepackage[colorinlistoftodos]{todonotes}
\usepackage{cleveref}
\usepackage{multirow}
\usepackage{multicol}
\crefname{section}{§}{§§}

\newtheorem{thm}{Theorem}
\newtheorem*{prop*}{Proposition}
\newtheorem{prop}{Proposition}

\theoremstyle{definition}
\newtheorem{defn}{Definition}

\newcommand{\cmark}{\ding{51}}
\newcommand{\xmark}{\ding{55}}

\begin{document}

\title{Pifthon: A Compile-Time Information Flow Analyzer For An Imperative Language}

\author{Sandip Ghosal}
\email{sandipsmit@gmail.com}
\orcid{0000-0002-3063-6080}
\affiliation{%
  \institution{Indian Institute of Technology, Bombay, India}
}

\author{R.K. Shyamasundar}
\email{shyamasundar@gmail.com}
\affiliation{%
  \institution{Indian Institute of Technology, Bombay, India}
}

\begin{abstract}
Compile-time information flow analysis has been a promising technique for protecting confidentiality and integrity of private data. In the last couple of decades, a large number of information flow security tools in the form of run-time execution-monitors or static type systems have been developed for programming languages to analyze information flow security policies. However, existing flow analysis tools lack in precision and usability, which is the primary reason behind not being widely adopted in real application development. In this paper, we propose a compile-time information flow analysis for an imperative program based on a \textit{hybrid} (mutable + immutable) labelling approach that enables a user to detect information flow-policy breaches and modify the program to overcome violations. We have developed an information flow security analyzer for a dialect of Python language, \textit{PyX},  called \textit{Pifthon}\footnote{https://github.com/pifthon/pifthon} using the said approach . The flow-analyzer aids in identifying possible misuse of the information in sequential PyX programs corresponding to a given \textit{information flow policy} (IFP).  \textit{Pifthon} has distinct advantages like reduced labelling overhead that ameliorates usability, covers a wide range of PyX programs that include termination-and progress-sensitive channels, in contrast to other approaches in the literature. The proposed flow analysis is proved to be sound under the classical non-interference property. Further, case study and experience in the usage of {\it Pifthon\/} are provided.

\end{abstract}



\keywords{language security, information flow control, python, mutable labels}
\maketitle

\section{Introduction}
\label{sec:introduction}
Information flow control (IFC) is concerned with either preventing or facilitating (il)legal information flow in a program according to a given flow policy. Denning's security lattice model \cite{denning1976lattice} provided the initial momentum for such a goal and has given rise to refinements of such an analysis to realize the needed precision using various new security models \cite{myers2000protecting,stefan2011disjunction,kumar2017complete}. A common feature of these security models is security \textit{labels} or \textit{classes} (used interchangeably) that are elements of a security lattice and used to specify confidentiality and integrity of program objects. A labelling function $\lambda$ maps subjects (stakeholders of a program in execution) and objects (variables, files) of a program to the respective security label from the lattice. Then, the necessary condition for the flow security of a program is: if there is information flow from object $x$ to $y$, denoted by $x\rightarrow y$, then the flow is secure if $\lambda(x)\textit{ can-flow-to }\lambda(y)$ in the lattice \cite{denning1977certification}. The condition is often referred to as the \textit{Information Flow Policy} (IFP). It has given rise to a  challenging area of certification of computer programs; a program is certified to be  information flow secure if there are no violations of IFP at any point during the execution. Recent work has proposed a number of certification mechanisms in the form of run-time monitors \cite{le2007automaton,askarov2009tight,austin2009efficient,stefan2011flexible,buiras2014dynamic} or statically checked type languages \cite{volpano_smith,abadi1999core,zdancewic2002secure_a,smith2001new,honda2002uniform,myers2001jif,pottier2002information,banerjee2002secure,simonet2003flow,broberg2013paragon} often augmented with run-time label checking support \cite{myers2001jif,zheng2007dynamic}.

The efficiency of a flow certification mechanism primarily depends on precise flow analysis. The precision could be defined as follows:  Let $F$ be the set of all possible information flows in a system, and $A$ be the subset of $F$ authorized by a given IFP. Let  $E$ be the subset of $F$ ``executable'' given the flow control mechanisms in operation. The system is \textit{secure} if $E\subseteq A$, i.e., all executable flows are authorized. A secure system is \textit{precise} if $E=A$, i.e., all authorized flows are executable \cite{robling1982cryptography}.

The precision of flow analysis is contingent on distinct labelling mechanisms adopted in the current literature \cite{myers2001jif,russo2010dynamic,stefan2011flexible,zheng2007dynamic,buiras2014dynamic}. The labelling mechanism usually refers to the process of binding subjects and objects of a program with a security class (statically or dynamically) from the lattice. Each of the existing compile-time and run-time certification mechanisms chooses an object labelling method ranging over the classes outlined in Table \ref{table:label_classification}.

\begin{table}[htb]
\centering
\caption{Classification of object labelling mechanisms}
\label{table:label_classification}
\begin{tabular}{l|l|l}
\hline
\begin{tabular}{l}
{\bf pc Label}$\rightarrow$\\ \hline
{\bf labelling}$\downarrow$
\end{tabular}
& {\bf Reset} & {\bf Monotonic}\\ \hline
{\bf Static} & $P_1$ & $P_2$\\ \hline
{\bf Dynamic} & $P_3$ & $P_4$\\ \hline
\end{tabular}
\end{table}

\begin{table*}[htb]
\centering
\small
\caption{Constraints generated by the labelling mechanisms discussed in Table \ref{table:label_classification}}
\label{table:constraints}
\begin{tabular}{l|l|l|l|l}
\hline
\textbf{Program} & \multicolumn{4}{|c}{\textbf{Constraints generated by:}}\\
\cline{2-5}
 & \multicolumn{1}{|c|}{$P_1$} & \multicolumn{1}{|c|}{$P_2$} & \multicolumn{1}{|c|}{$P_3$} & \multicolumn{1}{|c}{$P_4$}\\
\hline
1 \texttt{while} $x_0$: & \multicolumn{4}{|c}{$\lambda(x_i)=\underline{x_i}$, $\forall i \in 1\dots 5$}\\
\cline{2-5}
2 \hspace{3mm} $x_4 = x_1$ & $\underline{x_0}\oplus\underline{x_1}\leqslant\underline{x_4}$ & $\underline{x_0}\oplus\underline{x_1}\leqslant\underline{x_4}$ & $\underline{x_0}\oplus\underline{x_1}\leqslant\underline{x_4}$ & $\underline{x_0}\oplus\underline{x_1}\leqslant\underline{x_4}$\\
 & & $\underline{pc}=\underline{x_0}\oplus\underline{x_1}$ & $\underline{x_4}=\underline{x_0}\oplus\underline{x_1}\oplus\underline{x_4}$ & $\underline{pc}=\underline{x_0}\oplus\underline{x_1}$; $\underline{x_4}=\underline{x_0}\oplus\underline{x_1}\oplus\underline{x_4}$\\
\cline{2-5}
3 \hspace{3mm} $x_1 = x_2$ & $\underline{x_0}\oplus\underline{x_2}\leqslant\underline{x_1}$ & $\underline{x_0}\oplus\underline{x_1}\oplus\underline{x_2}\leqslant\underline{x_1}$ & $\underline{x_0}\oplus\underline{x_2}\leqslant\underline{x_1}$ & $\underline{x_0}\oplus\underline{x_1}\oplus\underline{x_2}\leqslant\underline{x_1}$\\
 & & $\underline{pc}=\underline{x_0}\oplus\underline{x_1}\oplus\underline{x_2}$ & $\underline{x_1}=\underline{x_0}\oplus\underline{x_1}\oplus\underline{x_2}$ & $\underline{pc}=\underline{x_0}\oplus\underline{x_1}\oplus\underline{x_2}$; $\underline{x_1}=\underline{x_0}\oplus\underline{x_1}\oplus\underline{x_2}$\\
\cline{2-5}
4 \hspace{3mm} $x_2 = x_3$ & $\underline{x_0}\oplus\underline{x_3}\leqslant\underline{x_2}$ & $\underline{x_0}\oplus\underline{x_1}\oplus\underline{x_2}\oplus\underline{x_3}\leqslant \underline{x_2}$ & $\underline{x_0}\oplus\underline{x_3}\leqslant\underline{x_2}$ & $\underline{x_0}\oplus\underline{x_1}\oplus\underline{x_2}\oplus\underline{x_3}\leqslant \underline{x_2}$\\
 & & $\underline{pc}=\underline{x_0}\oplus\underline{x_1}\oplus\underline{x_2}\oplus\underline{x_3}$ & $\underline{x_2}=\underline{x_0}\oplus\underline{x_2}\oplus\underline{x_3}$ & $\underline{pc}=\underline{x_0}\oplus\underline{x_1}\oplus\underline{x_2}\oplus\underline{x_3}$; $\underline{x_2}=\underline{x_0}\oplus\underline{x_1}\oplus\underline{x_2}\oplus\underline{x_3}$\\
\cline{2-5}
5 $x_1 = x_5$ & $\underline{x_5}\leqslant\underline{x_1}$ & $\underline{x_0}\oplus\underline{x_1}\oplus\underline{x_2}\oplus\underline{x_3}\oplus\underline{x_5}\leqslant\underline{x_1}$ & $\underline{x_5}\leqslant\underline{x_0}\oplus\underline{x_1}\oplus\underline{x_2}$ & $\underline{x_0}\oplus\underline{x_1}\oplus\underline{x_2}\oplus\underline{x_3}\oplus\underline{x_5}\leqslant\underline{x_0}\oplus\underline{x_1}\oplus\underline{x_2}$\\
 & & $\underline{pc}=\underline{x_0}\oplus\underline{x_1}\oplus\underline{x_2}\oplus\underline{x_3}\oplus\underline{x_5}$ & $\underline{x_1}=\underline{x_0}\oplus\underline{x_1}\oplus\underline{x_2}\oplus\underline{x_5}$ & $\underline{pc}=\underline{x_0}\oplus\underline{x_1}\oplus\underline{x_2}\oplus\underline{x_3}\oplus\underline{x_5}$;\\
 & & & & $\underline{x_1}=\underline{x_0}\oplus\underline{x_1}\oplus\underline{x_2}\oplus\underline{x_3}\oplus\underline{x_5}$\\
\hline
\end{tabular}
\end{table*}

\noindent A labelling mechanism that follows $P_1$ or $P_2$ evaluates immutable static labels for the objects prior to the computation. Whereas, a label of an object, as per $P_3$ or $P_4$, evolves to accommodate flows as the mechanism encounters during the computation, thus achieves \textit{flow-sensitivity} \cite{hunt2006flow}. Further, the classes  capture implicit flows by tracking progress of program counter ($pc$) whose label is either reset after every program statement ($P_1$, $P_3$) or updated monotonically ($P_2$, $P_4$). However, static labelling, i.e., $P_1, P_2$, is often too restrictive to recognize a benign program. Dynamic labelling, i.e., $P_3, P_4$, on the other hand, is also inappropriate since they may acknowledge all programs as secure. The experience motivates us to follow a hybrid labelling that would have a nice trade-off for static and dynamic labels. Moreover, both the labelling approaches lack in precision while computing labels for certain scenarios, for example, not reflecting the potential influences by forward and recurring backward information flows for an iterative statement. This often leads to information leaks due to the non-termination of programs or violation of non-interference property \cite{goguen1982security,volpano_smith}. We compare subtle precision in flow analysis captured by these labelling classes using a classic while-program structure in Section \ref{sec:constraints}.

In this work, we present a flow analysis that adopts a compile-time hybrid labelling with $pc$-monotonic that would have an excellent trade-off for mutable and immutable labels to realize acceptable precision. A set of system objects called \textit{global} that are sensitive to the outside world, given immutable labels. Whereas, the intermediate objects, called \textit{local}, have mutable labels that change dynamically during computation. Further, our labelling mechanism unrolls a loop statement for a finite time to capture the influences of backward information flow. Also, by increasing the $pc$ label monotonically, we capture the inter-statement forward information flow. Thus the mechanism paves the way for realizing termination-and progress-sensitive information leaks in sequential programs. Given the labelling approach, we develop a flow analyzer \textit{Pifthon} for an extended subset of Python language. The analyzer auto-generates final labels for the local objects from a given set of global objects and their immutable labels, provided the program adheres to the IFP at each program point.


Usually, enforcement of the global policy of non-interference ensures information flow only in the upward direction in the lattice. In practice, often system design demands information to flow downward direction too. For this purpose, the notion of  \textit{declassification} or \textit{downgrading} \cite{myers2000protecting,zdancewic2001robust,sabelfeld2003model,hicks2006languages,king2007automatic} has been used extensively under  rigid conditions so that confidentiality is not violated to the detriment of the usage. Downgrading eventually allows more subjects to be readers of the information. However, the addition of readers needs to be genuinely robust as it may reduce to pure discretionary access control that has severe consequences in a decentralized model. For this purpose, we are using the ``downgrading'' rule from \texttt{RWFM} model \cite{kumar2014realizing} that is robust from the specification perspective.

Main contributions of our work are:\\
$\bullet$ We propose a compile-time hybrid (static+dynamic) labelling approach that updates $pc$-label monotonically. Labels of global subjects/objects remain static and labels of local objects start with the least restrictive label but dynamically change to adjust flows during computation.\\
$\bullet$ We design an approach that unrolls iterative statements for finite times to analyze the recurring backward information flows. The approach, while updating $pc$-label monotonically, provides a solution to identify leaks through termination- and progress-sensitive channels.\\
$\bullet$ We develop \textit{Pifthon}-flow analyzer, for a dialect of Python language based on the proposed labelling mechanism. It generates final labels for local objects (including $pc$) for a given set of immutable labels of global objects of the program. It also identifies program points that could leak information.\\ 
$\bullet$ Our system uses a non-discretionary downgrading mechanism that does not require global knowledge of the system by the user.


\textbf{Structure of the paper:}  Characterization of labelling schemas is discussed in the Section \ref{sec:introduction} followed by the syntax and semantics of PyX in section \ref{sec:language}. 
In section \ref{sec:pifthon_characteristics}, we discuss characteristics of  \textit{Pifthon} followed by the labelling semantics of the tool in Section \ref{sec:labelling_semantics}.  Section \ref{sec:downgrading} describes the downgrading approach used. In section \ref{sec:experiences}, we discuss  user experience of \textit{Pifthon}. Section \ref{sec:related_work} provides a relative comparison of the approaches, followed by conclusion in Section \ref{conclusion}.

\section{Labelling Mechanisms}
\label{sec:constraints}
This section derives a relation among labelling mechanisms discussed in Section \ref{sec:introduction}. The relation highlights the contrast in precision measured in terms of the registered flow channels by each mechanism. Table \ref{table:constraints} shows the constraints generated according to the labelling schemes for a simple while-program. Constraints are written in terms of security conditions that should satisfy for an information flow at a program point. For instance, consider line 3, where the information flows in the direction $x_0\rightarrow x_2\rightarrow x_1$. The flow is secure only if it satisfies  $\underline{x_0}\oplus \underline{x_2}\oplus\underline{pc}\leqslant\underline{x_1}$, where $\underline{x_i}$ denotes the security label of $x_i$, $\underline{pc}$ denotes the label of the program counter, and $\oplus$ is a binary class combining operator (\textit{join}) that evaluates least upper bound (LUB) of two security classes. Thus, flow constraints according to schemes, e.g., $P_1$, $P_2$, at line 3 are as follows:\\
$\bullet$ Since $P_1$ resets $\underline{pc}$ to the least class \textit{bottom} or \textit{low} ($\bot$) after executing each statement, the security constraint is straightforward: $\underline{x_0}\oplus \underline{x_2}\leqslant\underline{x_1}$;\\
$\bullet$ As per $P_2$, since $pc$ has accessed the variables $x_0$ and $x_1$ till line 2, the $pc$ label is updated to $\underline{x_0}\oplus\underline{x_1}$. Therefore, the security constraint is evaluated as $\underline{x_0}\oplus \underline{x_1}\oplus\underline{x_2}\leqslant\underline{x_1}$.\\
Similarly, we obtain the constraints at each program point following the semantics of the schemes $P_1$, $P_2$, $P_3$ and $P_4$. 

We draw the following observations from Table \ref{table:constraints}: (i) scheme $P_1$ yields the smallest set of possible flows, misses to attain the inter-statement flows, for instance, information flow from the loop guard to the statement following the loop,  (ii) the schemes $P_2$ and $P_3$ fail to attain the complete set of global flows, and (iii) $P_4$ subsumes all the flows achieved by the other three schemes. We envision a Venn diagram shown in Figure \ref{binding_venn} that depicts relationships among the four labelling schemes. Consider the outermost circle depicts the set of all possible flows that could exist for a given program. Then $P_4$ outperforms the other labelling schemes in terms of precision. Nonetheless, it fails to include information flows that could occur because of multiple iterations of the loop, for example, $\underline{x_0}\oplus\underline{x_1}\oplus\underline{x_2}\oplus\underline{x_3}\leqslant\underline{x_4}$.

\begin{wrapfigure}{Ri}{4cm}
\centering
\includegraphics[width=4cm,height=4cm,keepaspectratio]{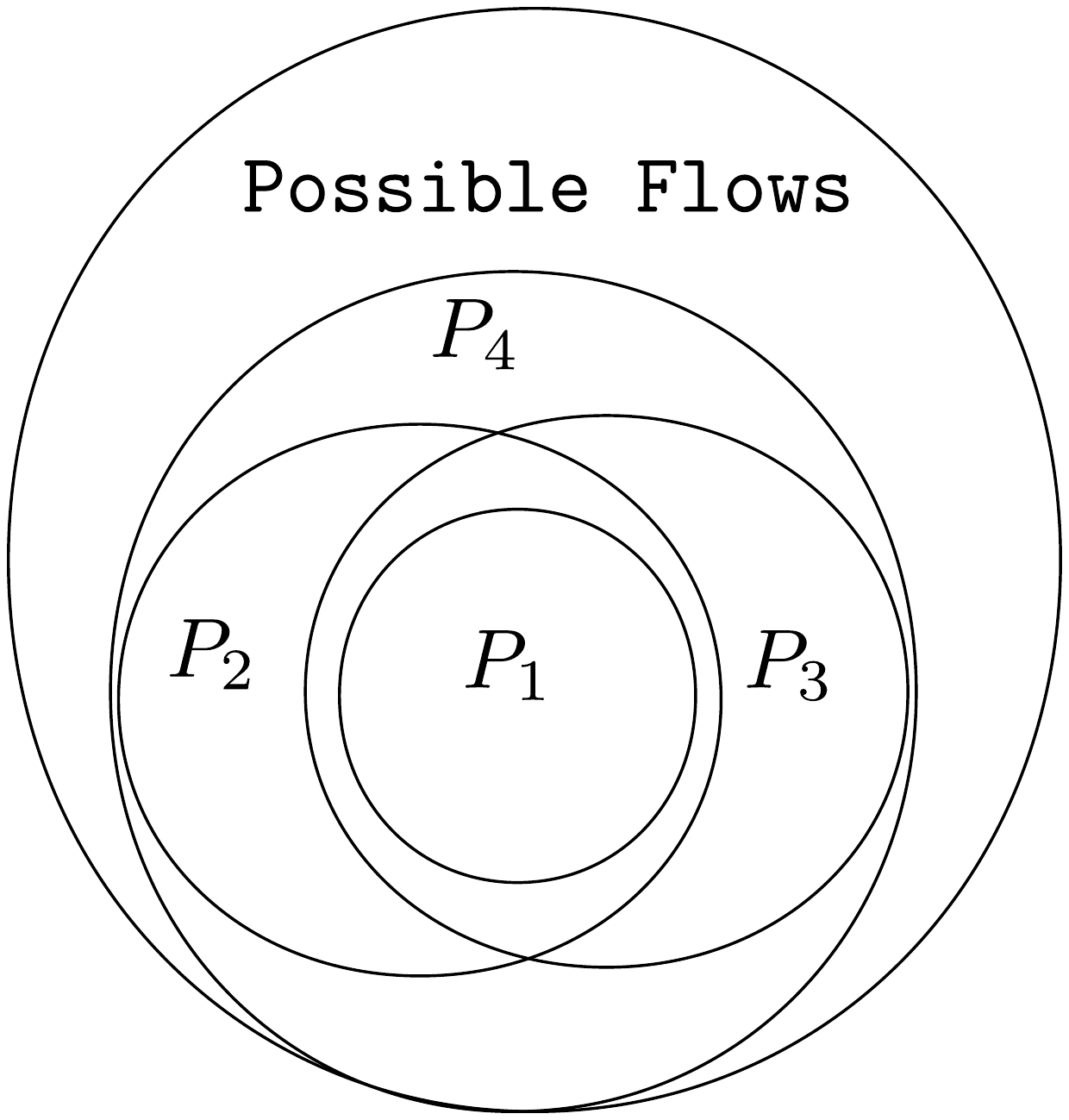}
\caption{Relation among the labelling mechanisms}
\label{binding_venn}
\end{wrapfigure}

We summarize the shortcomings of the discussed labelling schemes in terms of precision: (i) static labelling approaches such as $P_1, P_2$ are inappropriate as they are too restrictive, and might result in secure programs being incorrectly rejected as insecure or vice-versa (ii) dynamic labelling schemes $P_3, P_4$ are also inappropriate as they may end up allowing all programs as secure (iii) flow analysis of loop statements (terminating, non-terminating, abnormally terminating including exceptions) under the schemes that reset the label of the program counter, i.e., $P_1,P_3$, miss forward information flows, therefore fails to realize information leaks due to \textit{termination-}\cite{askarov2008termination},\textit{progress-sensitive}\cite{volpano1997eliminating} channels (e.g., \textit{copy5}, \textit{copy6}\cite{robling1982cryptography}), and (iv) schemes that monotonically increase the label of program counter, i.e., $P_2, P_4$, successfully capture the forward information flow, but miss the impact on object labels caused by recurring backward information flow.   

We introduce a semi-dynamic labelling (a subset of objects referred as \textit{global} have static labels) of objects in the program, that could be unrolled a finite number of times for iterative statements to capture both the forward information flow as well as the impact on the security labels due to backward label propagation.

In the remaining part of this section, we provide a vivid description of the notion of a security label using the Readers-Writers Flow Model (RWFM) that is succinct in specifying confidentiality and integrity policies associated with an information.

\subsection{\texttt{RWFM} Labels}
\label{subsec:labels}
We borrow a lattice-based information flow control model, i.e., Readers-Writers Flow Model (RWFM) \cite{kumar2017complete} for labelling subjects and objects and governing information flows in a program. In \texttt{RWFM}, a  subject or principal is a string representation of a source of authority such as user, process, also called as \textit{active} agent of a program responsible for information flow. On the other hand, objects are passive agents such as variables, files used for storing information. 

A \texttt{RWFM} label $l$ of a subject or object is a three-tuple $(s,R,W)$ ($s,R,W\in\text{ set of principals } P$), where $s$ represent the owner of the information and policy, $R$ denote the set of subjects (\textit{readers}) allowed to read the information, and $W$ refer to the set of subjects (\textit{writers}) who have influenced the information so far. The readers and writers set, respectively, specify the confidentiality and integrity policy associated with the information. Information from a source with a \texttt{RWFM} label $L_1$ can flow to an endpoint having \texttt{RWFM} label $L_2$ ($L_1\leqslant L_2$) if it does not violate the confidentiality and integrity policies already imposed by $L_1$. This \textit{can-flow-to} relationship is defined as below:

\begin{defn}[\textbf{Can-flow-to relation} ($\mathbf{\leqslant}$)]
Given any two \texttt{RWFM} labels $L_1=(s_1,R_1,W_1)$ and $L_2=(s_2,R_2,W_2)$, the \textit{can-flow-to} relation is defined as:
\[
\dfrac{R_1\supseteq R_2\quad W_1\subseteq W_2}{L_1\leqslant L_2}
\]
\end{defn}

The join ($\oplus$) and meet ($\otimes$) of any two \texttt{RWFM} labels $L_1=(s_1,R_1,W_1)$ and $L_2=(s_2,R_2,W_2)$ are respectively defined as
\[L_1\oplus L_2 = (-,R_1\cap R_2,W_1\cup W_2),\quad L_1\otimes L_2 = (-,R_1\cup R_2,W_1\cap W_2)\]

\noindent Then the set of \texttt{RWFM} labels $SC=P\times2^P\times 2^P$ forms a bounded lattice $(SC,\leqslant,\oplus,\otimes,\top,\bot)$, where $(SC,\leqslant)$ is a partially ordered set and $\top=(-,\emptyset,P)$, and $\bot=(-,P,\emptyset)$ are respectively the maximum and minimum elements.

\noindent{\textbf{Example of \texttt{RWFM} Labels and Operation:}}\\
Let a labelling function $\lambda$ maps an object $x$ to the label $(A,\{A\},\{A, B\})$ ($A,B\in P$), interpret as $A$ is the owner of the information in $x$; $A$ allows only herself to read the information; and $A, B$ have contributed to the information so far. Similarly, consider another object $y$ is mapped to the label $(B,\{A,B\},\{B\})$. Then, as per the \texttt{RWFM}, information flow from $y$ to $x$ is allowed ($\lambda(y)\leqslant\lambda(x)$) since $R_{\lambda(y)}\supseteq R_{\lambda(x)}$ and $W_{\lambda(y)}\subseteq W_{\lambda(x)}$. Further, a joining operation of the labels $\lambda(x)$ and $\lambda(y)$ shall be performed as follows:
\begin{align*}
R_{\lambda(x)}\cap R_{\lambda(y)} = \{A\}\quad W_{\lambda(x)}\cup W_{\lambda(y)}=\{A,B\}\\
\lambda(x)\oplus\lambda(y) = (-,\{A\},\{A,B\})
\end{align*}

\section{$\mathbf{PyX}$ Language: Syntax and Semantics}
\label{sec:language}
Here, we describe the syntax and semantics of PyX, and the labelling adopted in our analyzer.

A subset of Python language, extended with a command for explicit downgrading, constitute our core programming language called PyX. The language in general keeps to the syntax, alignment, and indentation followed in Python framework.

\subsection{Syntax}
\label{subsec:syntax}
Abstract syntax of PyX is shown in Figure \ref{pyx_syntax}. Two basic syntactic categories of the language are: expressions and statements. Let $c$ denote a statement, $n$, $x$, $p$ range over the set of literals, program variables and principals respectively. Expressions, $e$, are built up by variables, literals, boolean variables, binary and unary arithmetic operators, and function references.
\begin{figure}[htb]
\begin{tabbing}
$e$ ::=\=\hspace{0.5cm}$n$ $|$ $x$ $|$ $e$ op $e$  $|$ op $e$  $|$ $fcall$\\
$c$ ::=\=\hspace{0.5cm}{\bf \texttt{pass}} $|$ $x=e$ $|$ $c$ $[c..]$ $|$ {\bf \texttt{if}} $e$: $c_1$ {\bf \texttt{else:}} $c_2$ $|$ {\bf \texttt{while}} $e$: $c$ $|$ $fdef$\\
\>\hspace{0.5cm}$|$ $fcall$ $|$ {\bf \texttt{return}} $x$ $|$ {\bf \texttt{return}} $ downgrade$\\
$fdef$ ::=\=\hspace{0.5cm}{\bf \texttt{def} f}($\epsilon|x[,x\dots]$): $c$ \\
$fcall$ ::=\=\hspace{0.5cm}\textbf{f}($\epsilon|x[,x\dots]$) \\
$downgrade$ ::=\=\hspace{0.5cm}{\bf \texttt{downgrade}}($x$, $\{p,\dots\}$)\\
\end{tabbing}
\caption{Syntax of PyX}
\label{pyx_syntax}
\end{figure}

\noindent 
For simplicity, we consider non-recursive functions that follow  \textit{call-by-value} approach for parameter passing, where the function parameters are variables only. Further, each function contains no more than one return statement that occurs at function exit. The construct \textit{downgrade} is used to reduce the sensitivity of a variable,  and only occurs at the time of function return. It takes a variable, whose sensitivity is the subject of interest, and a set of principals as input. While downgrading sensitivity, it performs an inclusion of the given set of principals into the existing \textit{reader} set of the variable and returns a new label. The notion of readers will become  clear in the sequel. Without ambiguity, we identify string representation of a principal or subject by $A$ or `$A$' instead of ``$A$''.

\subsection{Operational Semantics}
\label{subsec:pyx_semantics}
Small-step operational semantics for PyX is given in Figure \ref{fig:pyx_semantics}. We capture changes in the memory state using finite domain functions $\phi$ and $\sigma$ ranging over the set of variable environment $\Phi$ and storage environment $\Sigma$ respectively. The mapping functions together establish a state of the transition system and defined as $\phi\in\Phi=\mathcal{X}\rightarrow\mathcal{M}$ and $\sigma\in\Sigma=\mathcal{M}\rightarrow\mathcal{V}$, where $\mathcal{X}$ denotes a set of variables or storage cells, $\mathcal{M}$ denotes address space, and $\mathcal{V}$ denotes set of all possible values. The set of configurations of the transition system is $(\Phi\times\Sigma\times S)\cup(\Phi\times\Sigma)$, where $S$ denotes the set of program statements. Small-step transition relation for expressions is denoted by $\rightarrow_e\subseteq(\Phi\times\Sigma\times e)\times\mathcal{V}$ and for statements is denoted by $\rightarrow_c\subseteq(\Phi\times\Sigma\times S)\times ((\Phi\times\Sigma\times S)\cup(\Phi\times\Sigma))$.


\begin{figure}[htb]
\begin{center}
\fontsize{7}{6}\selectfont
\begin{tabular}{cc}
\hline
 & \\
 (SKIP) & $\dfrac{}{\langle\phi,\sigma,\texttt{pass}\rangle\rightarrow_c\langle\phi,\sigma,\epsilon\rangle}$ \\
 & \\ 
 (ASSIGN) & $\dfrac{\langle\phi,\sigma,e\rangle\rightarrow_e \langle\phi,\sigma,e'\rangle}{\langle\phi,\sigma,x=e\rangle\rightarrow_c\langle\phi,\sigma,x=e'\rangle}$\\
 & \\
 & $\dfrac{\langle\phi,\sigma,e'\rangle\rightarrow_e v\hspace{2mm}l=alloc(\mathcal{M})\hspace{2mm}l\notin codom(\phi)\cup dom(\sigma)}{\langle\phi,\sigma,x=e'\rangle\rightarrow_c\langle\phi[x\mapsto l],\sigma[l\mapsto v]\rangle}$\\
 & \\
 & $\dfrac{\langle\phi,\sigma,e'\rangle\rightarrow_e v\hspace{2mm}\phi(x)=l\hspace{2mm}l\in codom(\phi)\cup dom(\sigma)}{\langle\phi,\sigma,x=e'\rangle\rightarrow_c\langle\phi,\sigma[\phi(x)\mapsto v]\rangle}$\\
 & \\
 (COMPOSE) & $\dfrac{\langle\phi,\sigma,c_1\rangle\rightarrow_c\langle\phi_1,\sigma_1,c_1'\rangle}{\langle\phi,\sigma,c_1 c_2\rangle\rightarrow_c\langle\phi_1,\sigma_1,c_1' c_2\rangle}$\\
 & \\
 & $\dfrac{\langle\phi,\sigma,c_1\rangle\rightarrow_c\langle\phi_1,\sigma_1,\epsilon\rangle}{\langle\phi,\sigma,c_1 c_2\rangle\rightarrow_c\langle\phi_1,\sigma_1,c_2\rangle}$\\
 & \\
 (IF) & $\dfrac{\langle\phi,\sigma,e\rangle\rightarrow_e\langle\phi,\sigma,e'\rangle}{\langle\phi,\sigma,\texttt{if } e:c_1\texttt{ else: }c_2\rangle\rightarrow_c\langle\phi,\sigma,\texttt{if } e':c_1\texttt{ else: }c_2\rangle}$\\
  & \\
  & $\dfrac{\langle\phi,\sigma,e'\rangle\rightarrow_e True}{\langle\phi,\sigma,\texttt{if } True:c_1\texttt{ else: }c_2\rangle\rightarrow_c\langle\phi,\sigma,c_1\rangle}$\\
    & \\
  & $\dfrac{\langle\phi,\sigma,e'\rangle\rightarrow_e False}{\langle\phi,\sigma,\texttt{if } False:c_1\texttt{ else: }c_2\rangle\rightarrow_c\langle\phi,\sigma,c_2\rangle}$\\
  & \\
  (WHILE) & $\dfrac{\langle\phi,\sigma,\texttt{if }e:c \texttt{ while }e:c\texttt{ else: }pass\rangle\rightarrow_c\langle\phi',\sigma',c'\rangle}{\langle\phi,\sigma,\texttt{while }e:c\rangle\rightarrow_c\langle\phi',\sigma',c'\rangle}$\\
  & \\
  (FDEF) & $\dfrac{}{\langle\phi,\sigma,\textbf{def } f(x):c\rangle\rightarrow_{fdef}\langle\phi[f\mapsto(x,c,\phi)],\sigma\rangle}$\\
  & \\
  (FCALL) & $\dfrac{\let\scriptstyle\textstyle\substack{\langle\phi,\sigma,x\rangle\rightarrow_e v\hspace{2mm}\phi_2=\phi_1[f\mapsto\phi(f)][x\mapsto l]\\\phi(f)=(x,c,\phi_1)\hspace{2mm}l\notin codom(\phi)\cup dom(\sigma)\\\langle\phi_2,\sigma[l\mapsto v],c\rangle\rightarrow_c \sigma'}}{\langle\phi,\sigma,f(x)\rangle\rightarrow_c \sigma'}$\\
  & \\
\hline
\end{tabular}
\end{center}
\caption{Small-step operational semantics for PyX}
\label{fig:pyx_semantics}
\end{figure}

Let $f$ be a non-recursive function that accepts none or a single variable as a parameter.  State transformation caused by a function ranging over environments $\Theta_0$ and $\Theta_1$ is given by  $\Theta_0 = \Sigma\rightarrow\Sigma$ for a parameter-less function and $\Theta_1=\Sigma\times\mathcal{V}\rightarrow\Sigma$ for a function with a single parameter. Then, we extend the definition of storage environment as follows: $\Phi = \mathcal{X}\rightarrow(\mathcal{M}+\Theta)$, where $\Theta = \Theta_0+\Theta_1$. Note that, we can easily extend the definition of $\Theta$ for multiple parameters. A function declaration modifies the environment by associating with the function identifier, declaration environment, and the body of the function. We represent this new environment by \textit{closure}, that is a subset of $(\mathcal{X}\times S\times \Phi)$. Thus, a function declaration is merely a binding for $f$ to the closure. Small-step transition relation for function declaration is denoted by $\rightarrow_{fdef}\subseteq(\Phi\times\Sigma\times fdef)\times(\Phi\times\Sigma)$.

\begin{table*}[htb]
\caption{Static labelling vs. Dynamic labelling}
\label{table:dynamic_labelling}
\begin{tabular}{p{1cm}|p{6cm}|p{6cm}}
\hline
\multicolumn{1}{c|}{\textbf{Program}} & \multicolumn{1}{c|}{\textbf{Static labelling}} & \multicolumn{1}{c}{\textbf{Dynamic labelling}} \\
\hline
$a=x$ & \multirow{6}{6cm}{Static labelling attempts to infer the label of $a$ beforehand in such a way that all the flows to and from $a$ are secure. If $a$ is labelled as $\underline{x}\oplus\underline{z}$ due to explicit flows from $x$ and $z$ to $a$, the constraint $\underline{x}\oplus\underline{z}\leqslant\underline{y}$ does not satisfy as the condition given is $\underline{z}\nleqslant\underline{y}$. As static labelling fails to label the local variable $a$, the program is insecure.} & In dynamic labelling local variable $a$ is initially labelled as public ($\bot$) but lifted to $\underline{x}$ so that $x$ can flow to $a$.\\
$y=a$ &  & The flow from $a$ to $y$ is allowed since the constraint $\underline{x}\leqslant\underline{y}$ is satisfied.\\
$a=z$ &  & Label of $a$ is updated to  $\underline{x}\oplus\underline{z}$ so that flow is allowed as $\underline{z}\leqslant\underline{x}\oplus\underline{z}$. Hence the program is flow-safe.\\
\hline
\end{tabular}
\end{table*}

\subsection{Labelling PyX Programs}
\label{subsec:labelling_program}
Labelling procedure starts by identifying  global and local objects in a given PyX program. Essentially, objects sensitive to the outside world are considered global, but the programmer may use her discretion to choose any variable as a global object. A set of static \texttt{RWFM} labels of the global objects goes as input for the analyzer along with the PyX program. In addition, a programmer is required to specify the highest security class a subject or process executing the program can access -- often referred to as subject's \textit{clearance} label.

Of course, if global labels were annotated in the program itself, as in existing tools like \cite{myers2001jif,broberg2013paragon,buiras2015hlio}, it would have been easy. But, annotating a source code with security labels is often found problematic due to several reasons \cite{johnson2015exploring}, e.g., changing the security policy may require modifying several annotations, manual effort in annotating large programs \cite{hammer2010experiences}, etc. We ameliorate the situation by maintaining a single JSON file that contains immutable labels of all the global variables. PyX program, together with its' JSON file, are provided as input to \textit{Pifthon}. The tool automatically labels local variables with initial mutable labels.

Given a PyX program, global objects, their immutable \texttt{RWFM} labels and a clearance, \textit{Pifthon} auto-generates final labels for the local variables.

\section{Design Of the Analyzer}
\label{sec:pifthon_characteristics}
In this section, we provide an overview of the design of our flow analyzer $\textit{Pifthon}$ whose primary features are:\\
(i) It adopts a hybrid labelling mechanism that combines mutable and immutable security labelling schemes;\\
(ii) It manages a single program counter ($pc$) label, which increases monotonically evaluating the least upper bound of all the accessed objects;\\
(iii) It iterates loop statements until the labels of target objects in the loop body saturate -- capturing  backward information flows that arise due to multiple iterations.\\
Each of the above tasks is illustrated below along with the differences with other existing IFC approaches.

While the majority of the literature follows static (immutable) labelling, e.g., Jif \cite{myers2001jif}, Paragon \cite{broberg2013paragon}, FlowCaml \cite{simonet2003flow}, example in Table \ref{table:dynamic_labelling} shows the importance of dynamic (mutable) labelling since the former approach is too restrictive, and rejects secure programs as insecure. Consider the program fragment where the security labels of variables $x, y$ and $z$ are denoted by $\underline{x}$, $\underline{y}$ and $\underline{z}$ respectively such that $\underline{z}\nleqslant\underline{y}$, and $a$ is a local variable. If the program is analyzed under static labelling, it would fix the label of $a$ as $\underline{x}\oplus\underline{z}$ due to explicit flows from $x$ and $z$ to $a$. This fails to satisfy the constraint $\underline{a}\leqslant\underline{y}$ in line 2 due to the given condition $\underline{z}\nleqslant\underline{y}$. However, the program never causes an information flow from $z$ to $y$, and must be considered as secure if $\underline{x}\leqslant\underline{y}$.

The above example shows that following a static labelling scheme could miss several secure programs. However, a purely dynamic labelling is also inappropriate as it could end up declaring all programs as flow secure. Hence, an ideal label-checking mechanism should have a hybrid labelling. \textit{Pifthon} follows a hybrid labelling approach where the local objects of a program are labelled dynamically, whereas the global objects have static labels. Thus, the analyzer addresses a nice trade-off for static and dynamic labels to realize acceptable precision and performance.

Often an ``implicit'' flow is found to occur while executing branch or iterative statement even if the statement does not get executed. The literature of language-based security introduces the notion of program counter ($pc$) label to capture the impact of such implicit flows. Traditionally, once the control exits the conditional or iteration statements, the $pc$ label is reset to its previous value, thus denoting that the variables in the condition expression no longer impact the current context. However, the approach is flow-insensitive and does not capture forward information flow while exiting the control statement. Table \ref{table:pc_monotonic} shows the differences in the evaluation of $pc$ label and its' impact in the analysis by the tools following such an approach $pc$-reset \cite{myers2001jif,simonet2003flow,zheng2007dynamic,broberg2013paragon} and $pc$-monotonic, i.e., \textit{Pifthon}.

\begin{table}[htb]
\centering
\caption{$pc$-reset ($P_1,P_3$) vs. $pc$-monotonic ($P_2,P_4$)}
\label{table:pc_monotonic}
\begin{tabular}{lccc}
  \hline
  \textbf{Program} & \textbf{Constraints} & \multicolumn{2}{c}{Label of $pc$ ($\mathbf{\lambda(pc)}$)}\\
  \cline{3-4}
   & & $pc$-reset & $pc$-monotonic \\
  \hline
  1\hspace{3mm}y=0 & $\lambda(pc)\oplus\bot\leqslant\lambda(y)$ & $\bot$ & $\bot$ \\
  2\hspace{3mm}whie $x==0$:  & $\lambda(pc)=\lambda(x)\oplus\lambda(pc)$ & $\underline{x}$ & $\underline{x}$\\
  3\hspace{6mm}pass & & $\underline{x}$ & $\underline{x}$\\
  4\hspace{3mm}y=1 & $\lambda(pc)\oplus\bot\leqslant\lambda(y)$ & $\bot$ & $\underline{x}$\\
  \hline
\end{tabular}
\end{table}
\noindent Consider labels of the variables $x$ and $y$ that are given by $\underline{x}$ and $\underline{y}$ respectively, such that $\underline{x}\nleqslant\underline{y}$. Note that, the program initializes the value of $y$ as 0 and the following loop diverges if $x$ is also 0. If $x$ is 1, the program terminates with the value of $y$ set to 1. Therefore, the program is insecure since there is an implicit information flow from $x$ to $y$ depending on the termination of
the program. However, existing tools that  reset  $pc$ label, do not realize the forward information flow while exiting from the loop statement. On the contrary, \textit{Pifthon} detects the insecurity by following the approach $pc$-monotonic as described in the later section.

Next, we demonstrate the efficacy of \textit{Pifthon} in capturing the backward implicit flow in a loop statement through the example shown in Table \ref{table:backward_flow}. In the example, the static labels of $x$ and $y$ are given as $\underline{x}$, $\underline{y}$ respectively such that $\underline{x}\nleqslant\underline{y}$. $z$ is the only local variable with the given initial mutable label $\bot$. Note that information might flow from $x$ to $y$ via $z$ due to  backward flow caused by the while loop. This would cause violation of the baseline property of non-interference \cite{goguen1982security,volpano_smith}. Existing literature that enforces non-interference in programming languages have sidestepped addressing backward information flow. In our approach, the label of $pc$ plays a crucial role in capturing both the forward and backward implicit flows. Let the initial label of $pc$ be \textit{low}, i.e., $\bot$ but monotonically increases as it reads objects throughout the computations. While executing statement $z=x$, the label of $pc$ becomes LUB of $pc$ label and the label of $x$, i.e., $\bot\oplus\underline{x}=\underline{x}$. Note that, the computation requires one more unrolling of the loop to discover the violation of non-interference at the statement $y=z$ since $\underline{z}\oplus\underline{pc}\nleqslant\underline{y}$. Ideally, in the absence of any potential information leak, the computation would continue unrolling an iterative statement until there are no more changes in the label of local variables and $pc$ due to the iteration. 

\begin{table}[htb]
\centering
\caption{Loop: Recurring backward implicit flow from $x$ to $y$ }
\label{table:backward_flow}
\begin{tabular}{lccc}
  \hline
  \textbf{Statement} & & \textbf{$1^{st}$ Iteration} & \textbf{$2^{nd}$ Iteration}\\
   & & $\lambda(pc)$ &  $\lambda(pc)$\\
  \hline
  1\hspace{3mm}\texttt{while} $True$: & & $\bot$ & $\underline{z}\oplus\underline{x}$\\
  2\hspace{6mm}\texttt{y = z} & & $\underline{z}$ & $\underline{z}\oplus\underline{x}$ \\ 
  3\hspace{6mm}\texttt{z = x} & & $\underline{z}\oplus\underline{x}$ & \\ 
  \hline
\end{tabular}
\end{table}

The above characteristics of \textit{Pifthon} improve the precision in realizing a maximum subset of flow channels while making the compile-time analysis flow-sensitive and termination-sensitive.
\section{Security Semantics of $\mathbf{Pifthon}$}
\label{sec:labelling_semantics}
In this section, we describe the labelling semantics of \textit{Pifthon}. The specification of a PyX program available to the flow analyzer comprises: (i) $P$ - set of stakeholders for the computation, (ii) $V$ - set of program variables, (iii) $G$ - set of global variables ($G\subseteq V$) labelled with the given immutable \textit{RWFM} labels, (iv) $L\subseteq V$ - set of local variables such that $L\cap G=\{\emptyset\}$, (v) $p\in P$ - the subject or the principal with the authority required for executing the program that is being analyzed, (vi) $cl$ - the highest label the subject executing a program can access, and (vii) a program $c$.

The above program specification is given to \textit{Pifthon} with the objective of identifying program point that could lead to information misuse. A program is said to \textbf{MISUSE} information if there exists a command at a program point\\
(i) where the security label of the information source is higher than $cl$, or\\
(ii) which, when executed could cause an information flow that violates the underlying \textit{can-flow-to} relation between the labels of the source and the endpoint.\\

\noindent{\textbf{Notation:}}
Let $S$ denote the set of all commands and the initial program execution environment is described by the following  functions written as $f: S\rightarrow 2^V$ -- returns a set of variables for a command of PyX.\\
$\bullet$ Var: returns the set of all variables associated with a command,\\
$\bullet$ SV: returns the set of all variables that are the sources of information flows (implicit/explicit) for a command,\\
$\bullet$ TV: returns the set of variables that are the targets of information flows for a given command.\\
Further, functions $A$, $R$, and $W$ are the projections of a \texttt{RWFM} label to its first (owner), second (readers), and third (writers) component, respectively.  Note that the owner part of a security label is only for downgrading operations as  described in the sequel. An evaluation state is defined by the tuple $\lbrack \sigma,\lambda\rbrack$, where $\sigma$ range over the storage environment $\Sigma$ and $\lambda$ is the projection from variable identifiers including $pc$ to the respective security labels of \texttt{RWFM} lattice $\mathcal{L}$.

\subsection{Semantics}
Figure \ref{semantics_initial} succinctly defines  the initial environment for computation. The programmer provides the initial immutable labels $\lambda(x)$ for the global variables, i.e.,  $\forall x\in G$ and a clearance  label $cl$. Variable $pc$ (\textit{program counter}) represents the current stage of the control flow. Local variables and $pc$ are initially labelled with mutable label $(p,\{*\},\{\})$, where $p$ is the executing subject. Program $\langle P,V,G,p:c \rangle$ is analyzed  in the  environment $\Gamma = (P,V,G,p)$ and execution state $\lbrack \sigma,\lambda\rbrack$ only if the subject $p$ is a permissible reader of all the global sources in the program.

\begin{figure*}[htb]
\begin{center}
\fontsize{7}{6}\selectfont
\begin{tabular}{c}
\hline
 \\
$p\in P\hspace{0.5cm}G\subseteq V= Var(c)\hspace{0.5cm}\forall x\in SV(c)\cap G\lbrack p\in R(\lambda(x))\rbrack$\\
 \\
$\dfrac{\forall x\in SV(c)\cap G\lbrack \lambda'(x)=\lambda(x)\rbrack\hspace{0.3cm}\forall x\in (SV(c)\cap L\cup \{pc\})\lbrack \lambda'(x)=(p,\{*\},\{\})\rbrack}{\langle P,V,G,p :c\rangle \rightarrow (P,V,G,p)\vdash\langle c,\sigma',\lambda'\rangle}$\\
\\
$\dfrac{(p\notin P)\vee(G\nsubseteq V)\vee(V\neq Var(c))\vee(\exists x\in SV(c)\cap G\lbrack p\notin R(\lambda(x)) \rbrack\vee\exists x\in SV(c)\lbrack \lambda(x)\nleqslant cl\rbrack)}{\langle P,V,G,p :c \rangle \rightarrow \langle \text{MISUSE}\rangle}$\\
\\
\hline
\end{tabular}
\caption{Semantics for initial execution environment}
\label{semantics_initial}
\end{center}
\end{figure*}

Below, we provide a semantics that formally describes generation of security labels for local variables,  $pc$, as well as checking of security constraints for global variables in each PyX command. Derivations are of the form \[\Gamma\vdash\langle c,\sigma,\lambda\rangle\downarrow v^l\] where command $c$ with initial memory mapping $\sigma$ and labelling $\lambda$, evaluates to $v^l$ in the execution environment $\Gamma$. $v^l$ denotes a value $v$ labelled with $l$, where values range over by $v$ --  integer or boolean or string. Labelling semantics of \textit{Pifthon} is nothing but small-step transition relation denoted by $\rightarrow\subseteq(S\times\Sigma\times\mathcal{L})\times((S\times\Sigma\times\mathcal{L})\cup\text{MISUSE})$. 

Program literals are immutably labelled with the least element of \texttt{RWFM} lattice, i.e., $(-,\{*\},\{\})$, where $*$ denote all the stakeholders for a given program. Computation of \texttt{pass} does not require any memory access; thus, it has no effect on the evaluation state.

For an arithmetic operation, $op$, in $(e_1 \texttt{ op } e_2)$, label of the outcome is evaluated by $l_1\oplus l_2$ where $e_1$ and $e_2$ separately yield labels $l_1$ and $l_2$ respectively.

\begin{figure}[htb]
\begin{center}
\fontsize{7}{6}\selectfont
\begin{tabular}{cl}
\hline
  & \\
(INT) & $\Gamma\vdash\langle n, \sigma,\lambda\rangle\downarrow\llbracket n\rrbracket^{(-,\{*\},\{ \})}$ \\
  & \\
 (BOOL) & $\Gamma\vdash \langle True,\sigma,\lambda\rangle \downarrow True^{(-,\{*\},\{ \})}$\hspace{0.3cm}$\Gamma\vdash \langle False,\sigma,\lambda\rangle \downarrow False^{(-,*,\{ \})}$\\
  & \\
 (VAR) & $\Gamma\vdash\langle x, \sigma,\lambda\rangle\downarrow \sigma(\phi(x))^{\lambda(x)}$ \\
  & \\
 (SKIP) & $\Gamma\vdash\langle pass,\sigma,\lambda\rangle\rightarrow \langle\epsilon,\sigma,\lambda\rangle$ \\
  & \\ 
 (ARITH) & $\dfrac{\let\scriptstyle\textstyle\substack{\Gamma\vdash\langle e_1,\sigma,\lambda\rangle\downarrow v_1^{l_1}\hspace{0.3cm}\Gamma\vdash \langle e_2, \sigma,\lambda\rangle\downarrow v_2^{l_2}\hspace{0.3cm}v=(v_1 \texttt{ op } v_2 )\\l=l_1\oplus l_2}}{\Gamma\vdash\langle e_1\texttt{ op }e_2, \sigma,\lambda\rangle\downarrow v^l}$\\
  & \\
\hline
\end{tabular}
\end{center}
\caption{Semantics for expressions \& arithmetic operations}
\label{semantics_expr}
\end{figure}

Execution of  $x=e$ is interpreted as follows: first, evaluate expression $e$ on the RHS by reading the values in the variables appearing in $e$; second, write the result in the variable on the LHS. Let evaluation of $e$ result with value $v$ and label $l$. Label $l$ is equal to LUB of all the labels of the variables that occur in $e$. Naturally, accessing all the variables requires the subject clearance higher than $l$, which should satisfy $l\leqslant cl$. Since we increase the $pc$ label while accessing the variables, the new label of $pc$ becomes $l_1=l\oplus\lambda(pc)$. Although modern compilers could optimize accessing variables, e.g., while evaluating $0*y$, we enforce the analyzer to record the labels of all the variables to uphold best security practices. For the second step, the label of the left-hand side target variable $x$ should be no less restrictive than $l_1$. If the target is a local variable, the label of $x$ is updated by $l_1\oplus\lambda(x)$. On the other hand, if $x$ is a global variable ($x\in G$), the analyzer does not update the target label but only checks if $l_1\leqslant\lambda(x)$ to satisfy IFP, which if violated indicates a misuse. Thus, the analyzer follows static labelling when the target is a global variable; whereas, it uses dynamic labelling by updating labels when the target is a local variable. 

The above discussion provides the essence of hybrid labelling approach. In fact, it also makes the analyzer flow-sensitive. Semantics is illustrated through the example in Table \ref{tbl:assignment} (cf. Appendix \ref{sec:pifthon_output} to see the output from \textit{Pifthon}). 
\begin{figure}[htb]
\begin{center}
\fontsize{7}{6}\selectfont
\begin{tabular}{cl}
\hline
  & \\
 (ASSIGN) & $\dfrac{\let\scriptstyle\textstyle\substack{\Gamma\vdash\langle e,\sigma,\lambda\rangle\downarrow v^l\hspace{0.3cm}x\in L\hspace{0.3cm}\sigma'=\sigma\lbrack \phi(x)\mapsto v\rbrack\\l_1=l\oplus\lambda(pc)\hspace{0.3cm}\lambda'=\lambda\lbrack x\mapsto l_1,pc\mapsto l_1\rbrack}}{\Gamma\vdash\langle x=e,\sigma,\lambda\rangle\rightarrow\langle\epsilon,\sigma',\lambda'\rangle}$\\
  & \\
  & $\dfrac{\let\scriptstyle\textstyle\substack{\Gamma\vdash\langle e,\sigma,\lambda\rangle\downarrow v^l\hspace{0.3cm}x\in G\hspace{0.3cm}l_1=l\oplus\lambda(pc)\hspace{0.3cm}l_1\leqslant\lambda(x)\\\sigma'=\sigma\lbrack \phi(x)\mapsto v\rbrack\hspace{0.3cm}\lambda'=\lambda\lbrack pc\mapsto l_1\rbrack}}{\Gamma\vdash\langle x=e,\sigma,\lambda\rangle\rightarrow\langle\epsilon,\sigma',\lambda'\rangle}$\\
  & \\
  & $\dfrac{\let\scriptstyle\textstyle\substack{\Gamma\vdash\langle e,\sigma,\lambda\rangle\downarrow v^l\hspace{0.3cm}x\in G\hspace{0.3cm}l_1=l\oplus\lambda(pc)\hspace{0.3cm}l_1\nleqslant\lambda(x)\\\lambda'=\lambda[pc\mapsto l_1]}}{\Gamma\vdash\langle x=e,\sigma,\lambda\rangle\rightarrow\langle \textbf{MISUSE}\rangle}$\\
  & \\
\hline
\end{tabular}
\caption{Semantics for assignment statement}
\label{semantics_assign}
\end{center}
\end{figure}

\begin{table}[htb]
\centering
\caption{Computing labels for an assignment statement}
\label{tbl:assignment}
\begin{tabular}{l|l}
\hline
\multicolumn{1}{c|}{$\mathbf{c :: x=y}$, $x$ \textbf{is Local}} & \multicolumn{1}{c}{$\mathbf{c :: x=y}$, $x$ \textbf{is Global}}\\
\hline
\multicolumn{1}{l|}{\textbf{Initial labels:}} & \multicolumn{1}{l}{\textbf{Initial labels:}}\\
$\lambda(x)=(S,\{*\},\{\})$ & $\lambda(x)=(A,\{A,S\},\{A\})$\\
$\lambda(y)=(B,\{B,S\},\{B\})$ & $\lambda(y)=(B,\{B,S\},\{B\})$\\
$\lambda(pc)=(S,\{*\},\{\})$ & $\lambda(pc)=(S,\{*\},\{\})$ \\
\hline
\multicolumn{1}{l|}{\textbf{Label computations:}} & \multicolumn{1}{l}{\textbf{Label computations:}} \\
$l=\lambda(y)$ & $l=\lambda(y)$ \\
$l_1=l\oplus\lambda(pc)$ & $l_1=l\oplus\lambda(pc)$; $\lambda(pc)=l_1$\\
$\lambda(x)=l_1$ & check if $l_1\leqslant\lambda(x)$: \\
$\lambda(pc)=l_1$ & \quad MISUSE as $l_1\nleqslant\lambda(x)$\\
\hline
\multicolumn{1}{l|}{\textbf{Labels after computation:}} & \multicolumn{1}{l}{\textbf{Labels after computation:}} \\
$\lambda(x)=(S,\{B,S\},\{B\})$ & $\lambda(pc)=(S,\{B,S\},\{B\})$ \\
$\lambda(pc)=(S,\{B,S\},\{B\})$ & \\
\hline
\end{tabular}
\end{table}

We could interpret a conditional statement, e.g., $\texttt{if }e:\text{ }c_1\texttt{ else: }c_2$, as causing information flow from the predicate $e$ to both the branches $c_1$ and $c_2$ irrespective of the branch that gets executed during run-time. Therefore, the analyzer performs the following operations: (i) it checks if the labels of global target variables satisfy the IFP, violating which indicates a potential misuse, (ii) updates the labels of local target variables to the most restrictive label that should satisfy the IFP irrespective of the path taken during the execution, and (iii) increases the label of $pc$ to the upper bound of all the labels of source variables appear in both the branches. 
\begin{figure}[htb]
\begin{center}
\fontsize{7}{6}\selectfont
\begin{tabular}{cl}
\hline
  & \\
 (IF) & $\Gamma\vdash\langle e,\sigma,\lambda\rangle\downarrow v^l$ \hspace{0.3cm} $\forall x\in SV(c)\lbrack\lambda'(x)=\lambda(x)\rbrack$\\
  & \\
  & $\dfrac{\let\scriptstyle\textstyle\substack{\forall x \in TV(c)\cap G\lbrack l\oplus\lambda(pc)\oplus cl\leqslant \lambda(x)\rbrack\\\forall x \in TV(c)\cap L\cup\{pc\}\lbrack\lambda'(x)=l\oplus\lambda(x)\rbrack}}{\Gamma\vdash\langle \texttt{if }e:\text{ }c_1\texttt{ else: }c_2, \sigma,\lambda\rangle\rightarrow\langle c_1,\sigma',\lambda'\rangle}$\\
  & \\
  & $\dfrac{\let\scriptstyle\textstyle\substack{\Gamma\vdash\langle e,\sigma,\lambda\rangle\downarrow v^l\hspace{0.3cm}l_1=l\oplus\lambda(pc)\\\exists x\in TV(c)\cap G\lbrack l_1\oplus cl\nleqslant\lambda(x)\rbrack\hspace{0.3cm}\lambda'=\lambda\lbrack pc\mapsto l_1\rbrack}}{\Gamma\vdash\langle\texttt{if }e:\text{ }c_1\texttt{ else: }c_2,\sigma,\lambda\rangle\rightarrow\langle\textbf{MISUSE}\rangle}$\\ 
  & \\
\hline
\end{tabular}
\end{center}
\caption{Semantics for selection statement}
\label{semantics_if}
\end{figure}


For execution of an iteration, e.g., $\texttt{while }e:\text{ }S$, the analyzer requires to accumulate influences of all the variables in the loop condition. Next, it checks if all the potential endpoints in the loop body satisfy the IFP irrespective of whether the loop is ever taken during the run-time. Although the semantics for iteration looks similar to conditional statements, there is a subtle difference: a loop could be unrolled multiple times, which causes a backward information flow that carries the influence of one iteration to subsequent iterations. \textit{Pifthon} provides a solution that unrolls the loop until all the endpoints in the loop body and $pc$ reach their highest label. While unrolling, the $pc$ label carries forward the influence of one iteration to the next one. Semantics of a loop statement is shown in Figure \ref{semantics_while}.
\begin{figure}[htb]
\begin{center}
\fontsize{7}{6}\selectfont
\begin{tabular}{cl}
\hline
  & \\
 (WHILE) & $\dfrac{\let\scriptstyle\textstyle\substack{\Gamma\vdash\langle e,\sigma,\lambda\rangle\downarrow v^l\hspace{0.3cm}l_1=l\oplus\lambda(pc)\\\forall x\in TV(c)\cap G\lbrack l_1\leqslant\lambda(x)\rbrack\\\forall x\in SV(c)\cap L\cup\{pc\}\lbrack\lambda'(x)=\lambda(x),\lambda'(pc)\mapsto l_1\rbrack\\\Gamma\vdash\langle S,\sigma,\lambda'\rangle\rightarrow\langle\epsilon,\sigma',\lambda''\rangle\hspace{0.3cm}\lambda'\neq\lambda''}}{\Gamma\vdash\langle\texttt{while }e:\text{ }c,\sigma,\lambda\rangle\rightarrow\langle \texttt{ while }e:\text{ }c,\sigma',\lambda''\rangle}$\\
  & \\
  & $\dfrac{\let\scriptstyle\textstyle\substack{\Gamma\vdash\langle e,\sigma,\lambda\rangle\downarrow v^l\hspace{0.3cm}l_1=l\oplus\lambda(pc)\\\forall x\in TV(c)\cap G\lbrack l_1\leqslant\lambda(x)\rbrack\\\forall x\in SV(c)\cap L\cup\{pc\}\lbrack\lambda'(x)=\lambda(x),\lambda'(pc)\mapsto l_1\rbrack\\\Gamma\vdash\langle c,\sigma,\lambda'\rangle\rightarrow\langle\epsilon,\sigma',\lambda''\rangle\hspace{0.3cm}\lambda'==\lambda''}}{\Gamma\vdash\langle\texttt{while }e:\text{ }c,\sigma,\lambda\rangle\rightarrow\langle \epsilon,\sigma',\lambda''\rangle}$\\
  & \\
  & $\dfrac{\let\scriptstyle\textstyle\substack{\Gamma\vdash\langle e,\sigma,\lambda\rangle\downarrow v^l\hspace{0.3cm}l_1=l\oplus\lambda(pc)\\\exists x\in TV(c)\cap G\lbrack l_1\nleqslant\lambda(x)\rbrack\hspace{1mm}\lambda'=\lambda\lbrack pc\mapsto l_1\rbrack}}{\Gamma\vdash\langle\texttt{while }e:\text{ }c,\sigma,\lambda\rangle\rightarrow\langle\textbf{MISUSE}\rangle}$\\
  & \\
\hline
\end{tabular}
\end{center}
\caption{Semantics for iteration statement}
\label{semantics_while}
\end{figure}

Consider the termination-sensitive program shown in Table \ref{table:pc_monotonic}, where the labels of global variables $x$ and $y$ are given as $(A,\{A,S\},\{A\})$, and $(B,\{B,S\},\{B\})$ respectively, such that $\lambda(x)\nleqslant\lambda(y)$. Note that, in the first iteration, the $pc$ label is updated at line 2, i.e., $\lambda'(pc)=\lambda(x)\oplus\lambda(pc) = (S,\{A,S\},\{A\})$. \textit{Pifthon} iterates the loop for the second time where $pc$ carries the influence of first iteration. Since there is no other variable present in the loop body, \textit{Pifthon} only checks for the saturation in the $pc$ label. Once the saturation is confirmed after second iteration \textit{Pifthon} terminates the loop and advances the control to the following statement. Note that, during actual execution the program might not terminate at all. The tool now flags a \textbf{MISUSE} of information as $\lambda'(pc)\nleqslant\lambda(y)$. \textit{Pifthon's} output for the program is shown in the Appendix \ref{sec:pifthon_output}. 

\begin{table}[htb]
\centering
\caption{Flow analysis for a while program (Cf. \cite{robling1982cryptography})}
\label{table:example_iteration}
\begin{tabular}{lcc}
  \hline
  \textbf{Statement} & \textbf{Constraints} & $\mathbf{\lambda(pc)}$\\
  \hline
  1\hspace{3mm}y=0 & $\lambda(pc)\oplus\bot\leqslant\lambda(y)$ & $(S,\{*\},\{ \})$\\
  2\hspace{3mm}whie $x==0$:  & $\lambda(pc)=\lambda(x)\oplus\lambda(pc)$ & $1^{st}$ Itr.: $(S,\{A,S\},\{ \})$\\
  3\hspace{6mm}pass & & $2^{nd}$ Itr.: $(S,\{A,S\},\{ \})$\\
  4\hspace{3mm}y=1 & $\lambda(pc)\oplus\bot\nleqslant\lambda(y)$ & \\
  \hline
\end{tabular}
\end{table}

Obvious question is:\textit{what would happen if the labels of loop variables never saturate?} In other words, \textit{will this process of unrolling ever terminate?} It is proved that the labelling mechanism for a \texttt{while} statement will terminate after a maximum of three iterations (cf. Appendix \ref{append:proposition}).

A sequence statement consists of independent statements $c_1,\dots,c_n$. Note that in the execution of $c_1 c_2$,  execution of statement $c_2$ is conditioned on the program control reaching the end point of $c_1$, and $c_2$ executes only after $c_1$. Taking these observations into account, labelling $c_2$ shall be performed in the program context in which the labelling of $c_1$ is already accomplished i.e., the security label of $pc$ at the end of $c_1$ should be considered for the execution of $c_2$. 

\begin{figure}[htb]
\begin{center}
\fontsize{7}{6}\selectfont
\begin{tabular}{lll}
\hline
 (COMPOSE) & $\dfrac{\Gamma\vdash\langle c_1,\sigma,\lambda\rangle\rightarrow\langle c_1^{'},\sigma',\lambda'\rangle}{\Gamma\vdash\langle c_1 c_2,\sigma,\lambda\rangle\rightarrow\langle c_1^{'}c_2,\sigma',\lambda'\rangle}$ & $\dfrac{\Gamma\vdash\langle c_1^{'},\sigma,\lambda\rangle\rightarrow\langle\epsilon,\sigma',\lambda'\rangle}{\Gamma\vdash\langle c_1 c_2,\sigma,\lambda\rangle\rightarrow\langle c_2,\sigma',\lambda'\rangle}$\\
\hline
\end{tabular}
\end{center}
\caption{Semantics for sequence statement}
\label{semantics_seq}
\end{figure}

One can write statements like $h=h$ followed by $l=l$ ($h$ and $l$ are global and denote a \textit{high} and \textit{low} variables respectively) somewhere in the program. In such cases, the platform may indicate ``MISUSE'' as it fails to satisfy the constraint $h\oplus\underline{pc}\leqslant l$. We ignore such corner cases and leave the onus of correcting the code on the programmer.

A function call $f(y_1,\dots,y_n)$ causes information flow from function arguments $y_1,\dots,y_n$ to corresponding parameters $x_1,\dots,x_n$. As PyX follows the \textit{call-by-value} parameter passing mechanism, the parameters act as local variables within the scope of the function body. Therefore, the values and labels of these parameters are initialized with the values and labels evaluated for corresponding arguments. Next, \textit{Pifthon} computes the function body with the $pc$ initialized with mutable label $(p_1,\{*\},\{\})$, where $p_1$ is the subject executing the function. Note that a programmer could provide a clearance label for the scope of the function, which would then be treated as upper bound for $p_1$.
\begin{figure}[htb]
\begin{center}
\fontsize{7}{6}\selectfont
\begin{tabular}{cl}
\hline
  & \\
 (FCALL) & $\text{arg\_list}=\{y_1,\dots,y_n\}$ \hspace{0.3cm} $\text{param\_list}=\{x_1,\dots,x_n\}$ \\
  & \\
  & $\dfrac{\let\scriptstyle\textstyle\substack{\underset{i\in 1\dots n}{\forall}\Gamma\vdash\langle y_i,\sigma,\lambda\rangle\downarrow v_i^{l_i}\hspace{0.3cm}\sigma'=\sigma\lbrack \underset{i\in 1\dots n}{\forall}\phi(x_i)\mapsto v_i\rbrack\\\lambda = \lambda_{init}\hspace{0.3cm}\lambda'=\lambda\lbrack\underset{i\in 1\dots n}{\forall}x_i\mapsto l_i, pc\mapsto\underset{i\in 1\dots n}{\oplus}l_i\rbrack}}{\Gamma\vdash\langle\textbf{\texttt{f}}(y_1,\dots,y_n),\sigma,\lambda\rangle\rightarrow\langle\textbf{\texttt{f-body}},\sigma',\lambda'\rangle}$\\
  & \\
\hline
\end{tabular}
\end{center}
\caption{Semantics of PyX function call}
\label{semantics_function}
\end{figure}

Since function parameters are purely local, any change in the labels during the evaluation process does not affect corresponding arguments. However, it is necessary to have a \texttt{return} statement at the function exit to transfer the changes to the caller.

A \texttt{return} statement causes information to flow from the callee to the caller routine. Therefore, while returning a variable $x$, the returned label must incorporate the influence of the callee function. Further, consider $p$ to be the subject executing the caller. Then, $p$ must be an authorized reader of the returned label. We derive the security semantics for a return statement from the above conditions: (i) if $x$ is a local variable then the returned label is evaluated as $\lambda(x)\oplus\lambda(pc)$; (ii) if $x$ is a global variable, then return label $\lambda(x)$ if the following condition is satisfied: $\lambda(pc)\leqslant\lambda(x)$; if not,  it  flags a message for information misuse, and (iii) check if subject $p$ is a valid reader of the returned label. Violating the last condition paves the way for a controlled downgrading; this is discussed  in the next section.   
\begin{figure}[htb]
\begin{center}
\fontsize{7}{6}\selectfont
\begin{tabular}{cl}
\hline
  & \\
 (RETURN) & $\dfrac{\let\scriptstyle\textstyle\substack{p\in P\hspace{0.3cm}x\notin G\hspace{0.3cm}l=\lambda(x)\oplus\lambda(pc)\hspace{0.3cm}p\in R(l)\\\lambda'=\lambda\lbrack x\mapsto l, pc\mapsto l\rbrack}}{\Gamma\vdash\langle\texttt{return}\text{ }x,\sigma,\lambda\rangle\rightarrow\langle\epsilon,\sigma,\lambda'\rangle}$\\
  & \\
  & $\dfrac{\let\scriptstyle\textstyle\substack{p\in P\hspace{0.3cm}x\in G\hspace{0.3cm}\lambda(pc)\leqslant\lambda(x)\hspace{0.3cm}l=\lambda(x)\hspace{0.3cm}p\in R(l)\\\lambda'=\lambda\lbrack pc\mapsto l\rbrack}}{\Gamma\vdash\langle\texttt{return}\text{ }x,\sigma,\lambda\rangle\rightarrow\langle\epsilon,\sigma,\lambda'\rangle}$\\
  & \\
  & $\dfrac{\let\scriptstyle\textstyle\substack{x\in G\hspace{0.3cm}\lambda(pc)\nleqslant\lambda(x)}}{\Gamma\vdash\langle\texttt{return}\text{ }x,\sigma,\lambda\rangle\rightarrow\langle\textbf{MISUSE}\rangle}$\\
  & \\
\hline
\end{tabular}
\end{center}
\caption{Semantics of PyX return statement}
\label{semantics_return}
\end{figure}

\begin{figure*}[htb]
\begin{center}
\fontsize{7}{6}\selectfont
\begin{tabular}{lcc}
  (DOWNGRADE) & \multicolumn{2}{l}{$\dfrac{ \let\scriptstyle\textstyle\substack{p,p_1\in P \hspace{2mm} x\notin G \hspace{2mm}l=\lambda(pc)\oplus\lambda(x)\hspace{2mm}A(l)=p\wedge(W(l)=\{p\}\vee\{p_1\}\in W(l))\\ R_1=R(l)\cup\{p_1\}\hspace{2mm}\hspace{2mm}\lambda' =\lambda\lbrack x\mapsto (p,R_1,W(l)),pc\mapsto l\rbrack}}{\Gamma\vdash\langle\texttt{return downgrade($x$,$\{p_1\}$)},M,\lambda\rangle\rightarrow\langle\epsilon,M,\lambda'\rangle}$}\\
  & & \\
  & \multicolumn{2}{l}{$\dfrac{\let\scriptstyle\textstyle\substack{p,p_1\in P\hspace{2mm}x\in G\hspace{2mm}\lambda(pc)\leqslant\lambda(x)\hspace{2mm}l=\lambda(x)\hspace{2mm}A(l)=p\wedge(W(l)=\{p\}\vee\{p_1\}\in W(l))\\ R_1=R(l)\cup\{p_1\}\hspace{2mm}\hspace{2mm}\lambda' =\lambda\lbrack x\mapsto (p,R_1,W(l)),pc\mapsto l\rbrack}}{\Gamma\vdash\langle\texttt{return downgrade($x$,$\{p_1\}$)},M,\lambda\rangle\rightarrow\langle\epsilon,M,\lambda'\rangle}$}\\
  & & \\
  & \multicolumn{2}{l}{$\dfrac{x\notin G\hspace{2mm}l=\lambda(pc)\oplus\lambda(x)\hspace{2mm}p_1\notin P\vee A(l)\neq p\vee p_1\notin W(l)\vee W(l)\neq\{p\}\hspace{2mm}\lambda'=\lambda\lbrack pc\mapsto l\rbrack}{\Gamma\vdash\langle\texttt{return downgrade($x$,$\{p_1\}$)},M,\lambda\rangle\rightarrow\langle\textbf{MISUSE}\rangle}$}\\
  & & \\
  & \multicolumn{2}{l}{$\dfrac{\let\scriptstyle\textstyle\substack{x\in G\hspace{2mm}l=\lambda(x)\hspace{2mm}\lambda'=\lambda\lbrack pc\mapsto l\rbrack \\p_1\notin P \vee\lambda(pc)\nleqslant\lambda(x)\vee A(l)\neq p\vee p_1\notin W(l)\vee W(l)\neq\{p\}}}{\Gamma\vdash\langle\texttt{return downgrade($x$,$\{p_1\}$)},M,\lambda\rangle\rightarrow\langle\textbf{MISUSE}\rangle}$}\\
\end{tabular}
\end{center}
\caption{Semantics for downgrading}
\label{down_semantics}
\end{figure*}
\subsection{Soundness Of labelling Semantics}
\label{sec:non-interference}
We show that the labelling semantics \textit{Pifthon} is sound with respect to the classical definition of non-interference \cite{volpano_smith,2011_survey}. To meet the objective, first, we define the observation of an attacker who can inspect the data up to a label $\delta$.


\begin{defn}[Observation]
Given an environment $\Gamma$, a possible observation by an attacker labelled with $\delta$ is defined as a set of variables: \[Obs(\Gamma,\delta)=\{x\mid x\in TV\cup SV,\lambda(x)\leqslant\delta\}\]
\end{defn}

Let $\langle c,\Gamma\rangle\Downarrow_{\delta} o$ denote the execution of  statement $c$ in environment $\Gamma$ producing an observation $o$ with respect to label $\delta$. Two observations $o_1$ and $o_2$ are \textit{indistinguishable} with respect to any label $\delta$, denoted $o_1\simeq_{\delta} o_2$ if $\forall x\in o_1,\exists x\in o_2[\lambda(x)\leqslant\delta]$. Further, assume a \textit{low-equivalence} relation $\approx_{\delta}$ on two program statements $S$ and $S'$ if they differ only in high-security variables. Label $\delta$ defines ``high-security'' as follows: a variable tagged with any security label $h$ is called a high-security variable if it does not satisfy $h\leqslant\delta$. The attacker can inspect only data protected by $\delta$.

Non-interference says that any two low-equivalent statements are non-interfering as long as they produces indistinguishable observations. We write the condition for non-interference (NI) from the perspective of an attacker as shown below:
\[\text{If }S\approx_{\delta}S'\text{ then }(\langle S,\Gamma\rangle\Downarrow_{\delta}o\wedge\langle S',\Gamma\rangle\Downarrow_{\delta}o')\implies o\simeq_{\delta}o'\]
\noindent The following theorem establishes the soundness of \textit{Pifthon} under the above condition of non-interference.
\begin{thm}[Soundness with NI]
In a given environment $\Gamma$, if any two low-equivalent statements $S$ and $S'$ of a program are successfully labelled by \textit{Pifthon}, then the statements are non-interfering with respect to attacker's label $\delta$.
\label{thm:ni_sound}
\end{thm}
\begin{proof}
Let  $\delta=Low$. Let $x$ be a global target variable present in $S$ and $S'$. Let us assume evaluation relations $\langle S,\Gamma\rangle\Downarrow_{Low} o$ and $\langle S',\Gamma\rangle\Downarrow_{Low} o'$. Then $x\in o$ only if $\lambda(x)=Low$. Note that $S\approx_{Low}S'$, a potential update in the label of $x$ due to changes in high-security variable in $S'$ is guarded by  $l\oplus\lambda(pc)\oplus cl\leqslant\lambda(x)$. Since $\lambda(x)=Low$, and \textit{Pifthon} has labelled the variables successfully, we  immediately infer that $l=\lambda(pc)=cl=Low$ and $x\in o'$. This establishes the fact that $o\simeq_{Low}o'$. Hence the theorem follows.  
\end{proof}

\section{Downgrading}
\label{sec:downgrading}
Quite often, it is required to relax the confidentiality level and reveal some level of information to specific stakeholders for the successful completion of the transaction. For this purpose, the notion of \textit{Declassification} or \textit{Downgrading} needs to be captured either implicitly or explicitly. We illustrate the use of downgrading through the canonical example of a password update program \cite{myers2011attacker}. 

Consider the function \texttt{Password\_Update} that accepts two parameters: a new password ($new\_pwd$) and guessed password ($guess\_pwd$). The new password replaces the existing password ($pwd\_db$) if it matches with the guessed password. The function returns a boolean variable $result$ that is \textit{True} if there is a match or \textit{False} otherwise. A programmer could identify  variables $pwd\_db$, $guess\_pwd$ and $new\_pwd$  as global and $result$ as local. Let subject $A$ invoke the function to update the password retained by subject $B$. The desired security properties for the function are: (i) $guess\_pwd$ and $new\_pwd$ are readable by both subjects $A$ and $B$ but influenced only by $A$, and (ii) $pwd\_db$ is readable only by subject $B$ and influenced by subjects $A,B$. Next, the PyX program for \texttt{Password\_Update} function (shown in Figure \ref{fig:pwdupdate}) and \texttt{RWFM} labels of the global variables derived from the above security properties (shown in Table \ref{table:pwd_update_labels}) are given to \textit{Pifthon} for flow analysis (refer to Appendix \ref{sec:pwd_update_labelling}).

\begin{table}[htb]
\caption{Clearance/Immutable labels of subjects/global obj}
\label{table:pwd_update_labels}
\begin{tabular}{c|c}
\hline
    \textbf{Subject} & \textbf{Clearance Labels}  \\
\hline
     $A$ & $(A,\{A\},\{A,B\})$ \\
     $B$ & $(B,\{B\},\{A,B\})$\\
\hline
    \textbf{Object} & \textbf{\texttt{RWFM} Labels} \\
\hline
    $guess\_pwd$ & $(A,\{A,B\},\{A\})$\\
    $new\_pwd$ & $(A,\{A,B\},\{A\})$\\
    $pwd_db$ & $(B,\{B\},\{A,B\})$ \\
\hline
\end{tabular}
\end{table}
\begin{figure}[htb]
\begin{lstlisting}[frame=single]
def Password_Update(new_pwd, guess_pwd):
	result = False
	if pwd_db == guess_pwd:
		pwd_db = new_pwd
		result = True
	return result
	
guess_pwd = 'oldpwd'
new_pwd = 'mypwd'
success = Password_Update(new_pwd, guess_pwd)
\end{lstlisting}
\caption{PyX program for \texttt{Password\_Update} function}
\label{fig:pwdupdate}
\end{figure}




Observe that $result$ becomes a sensitive object by the time the function returns its' value. In the context of MLS, passing sensitive data (i.e., $result$) to a less-sensitive entity (i.e., subject $A$) would lead to an IFP violation as shown in Figure \ref{fig:downgrade_error}.

\begin{figure}[htb]
    \centering
    \includegraphics[width=8.5cm, height=4cm, keepaspectratio]{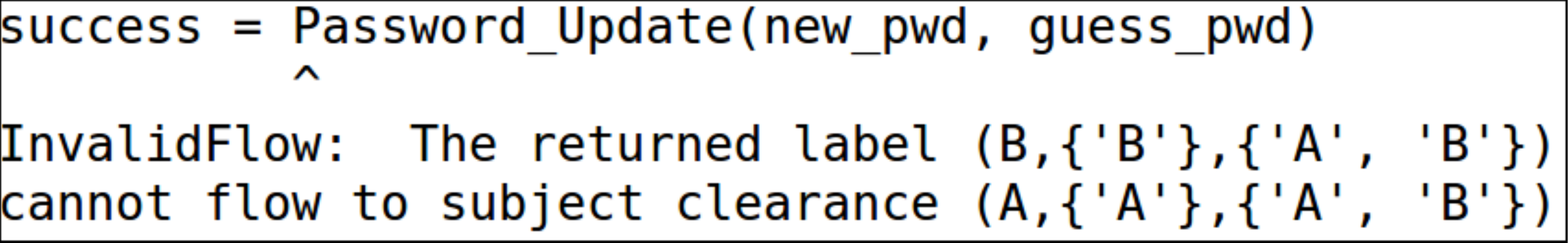}
    \caption{\textit{Pifthon} throwing error}
    \label{fig:downgrade_error}
\end{figure}

This demands a controlled \textit{downgrading}. There are two possibilities for introducing downgrading at the point of returning the value: \textit{first}, have an assertion that ensures downgrading explicitly, and \textit{second}, perform downgrading implicitly at the function return. Pifthon follows the former approach where the expression \texttt{downgrade($x$,$\{p_1,\dots\}$)} enables the programmer to specify a set of new readers.  While the existing IFC literature primarily follows the downgrading mechanism of \cite{myers2001jif,stefan2011disjunction} which is similar to discretionary, \textit{Pifthon} follows the downgrading rule of \texttt{RWFM} where downgrading is limited to the principals who had influenced the information at some earlier stages of information flow -- thus providing specific constraints:\\
The construct \texttt{downgrade($x$,$\{p_1\}$)} executed by $p$ adds $p_1$ to the readers list of label of $x$ only if:
\begin{itemize}
\item $p$ is the owner of the label and sole contributor of $x$ then $p$ may add any principal, e.g., $p_1$ into the readers set of $x$,
\item $p_1$ is one of the sources that influenced $x$ earlier then $p_1$ may be added to the readers list of $x$. 
\end{itemize}

As per the semantics of downgrading shown in Figure \ref{down_semantics},  \texttt{downgrade($x$,$\{p_1\}$)} executed by  principal $p$ relabels variable $x$ by adding reader $p_1$ into the current set of readers on performing the checks described above. Downgrade mechanism above  enforces a constraints on drop in confidentiality.


Thus, explicit downgrading \texttt{return downgrade(result, \{`A'\})} enables principal $A$ to access the outcome of the computation in password update program. 
\section{Pifthon Experiences}
\label{sec:experiences}

Here, we illustrate an effective use of \textit{Pifthon} in the analysis of Man-in-the-Middle (MITM) attack in Diffie-Hellman (DH) cryptographic key exchange protocol \cite{diffie1976new}. 

\noindent{\textbf{Case Study:}} Consider two entities $A$ and $B$ generating a shared private key $K$ using DH protocol.  Let $I$ be the intruder, and $S$ denote the system executing the program. A MITM attack is an outcome of two interleaving runs of DH protocol. In the first run, $I$ initiates a session with $A$ impersonating $B$. In the second run, $I$ further initiates a session with $B$ impersonating $A$. We instrument the attack in a PyX code and feed the program to \textit{Pifthon} along with the static labels of the global variables and subject clearances of functions given in Table \ref{tab:label_dh_protocol} . Labels comply with the security specification of the protocol. Suppose the subject executing the program has clearance label $(S,\{S\},\{A,B,S,I\})$. Then, \textit{Pifthon} performs the following label computation that prevents MITM attack (for simplicity, the network has been ignored):\\
(1) $A$ and $B$ agree on two public prime numbers $g$ and $p$.\\
(2) $A$ chooses a large random number $a$ as her private key and computes $m_a = g^a\text{ mod }p$ that evaluates the label of $m_a$ as $(A,\{A,S\},\{A,B\})$.\\
(3) Since the label does not include $B$ in the readers set, we downgrade the label of $m_a$ from $(A,\{A,S\},\{A,B\})$ to $(A,\{A,B,S\},\{A,B\})$ before sending the message to $B$.\\
\hspace{1cm}(3a) $I$ chooses a random number $i$ as her private key to evaluate $m_i=g^i\text{ mod }p$, which she sends to $A$ impersonating $B$. Message $m_i$ obtains the label $(I,\{I,S\},\{A,B,I\})$, and further downgraded  to $(I,\{I,A,S\},\{A,B,I\})$ before sending it to $A$.\\
(4) Similarly, $B$ chooses a number $b$ to be her private key and computes $m_b = g^b\text{ mod }p$. The label of $m_b$ is downgraded to $(B,\{A,B,S\},\{A,B\})$ before sending it to $A$.\\
\hspace{1cm}(4a) $I$ computes $m_i=g^i\text{ mod }p$ and downgrades the label of $m_i$ to $(I,\{I,B,S\},\{A,B,I\})$ before sending it to $B$ impersonating $A$.\\
(5) Next, $A$ computes shared private key $K_{AI}$ by evaluating $(g^i)^a\text{ mod }p$. $K_{AI}$ obtains  label $(A,\{A,S\},\{A,B,I\})$.\\
(6) Similarly, $B$ obtains shared private key $K_{BI}=(g^i)^b\text{ mod }p$ that receives label $(B,\{B,S\},\{A,B,I\})$.\\

Thus, any attempt to communicate with $B$ would require $A$ to encrypt the message using $K_{AI}$ and subsequently downgrade the label of the encrypted message to $B$. Since $A$ downgrades the label only to $B$, naturally intruder $I$ would fail to access the message. Furthermore, an attempt to intercept communication by the intruder even at an early stage (3(a) or 4(a)) will fail as messages $m_a$ and $m_b$ are downgraded, particularly for $A$ and $B$ respectively.  

\begin{table}[htb]
\small
\caption{Given \texttt{RWFM} labels for Diffie-Hellman-MITM}
\label{tab:label_dh_protocol}
\begin{minipage}[t]{4cm}
\centering
\begin{tabular}{c|c|c|c}
\hline
 Objects & \multicolumn{3}{c}{\texttt{RWFM} labels}\\
\cline{2-4}
 & Owner & Readers & Writers \\
\hline
 $g$ & $S$ & $\{*\}$ & $\{A,B\}$\\
 $p$ & $S$ & $\{*\}$ & $\{A,B\}$\\
 $a$ & $A$ & $\{A,S\}$ & $\{A\}$\\
 $b$ & $B$ & $\{B,S\}$ & $\{B\}$\\
 $i$ & $I$ & $\{B,S\}$ & $\{B\}$\\
\hline 
\end{tabular}
\end{minipage}
\hspace{2cm}
\begin{minipage}[t]{6cm}
\centering
\begin{tabular}{c|c|c|c}
\hline
 Functions & \multicolumn{3}{c}{\texttt{RWFM} labels} \\
\cline{2-4}
 & Owner & Readers & Writers \\
\hline
 creating\_m\_a & $A$ & $\{A,S\}$ & $\{A,B,S\}$\\
 creating\_k\_ai & $A$ & $\{A,S\}$ & $\{A,B,S\}$\\
 creating\_m\_b & $B$ & $\{B,S\}$ & $\{A,B,S\}$\\
 creating\_k\_bi & $B$ & $\{B,S\}$ & $\{A,B,S\}$\\
 creating\_m\_i & $I$ & $\{I,S\}$ & $\{A,B,I,S\}$\\
\hline 
\end{tabular}
\end{minipage}
\end{table}
We have also used \textit{Pifthon} for building a spectrum of security-critical benchmark applications with varying security properties and sizes. Table \ref{table:tried_programs} lists a few well-known security applications and cryptography protocols that we have developed on our secure platform. Thus, we have tested \textit{Pifthon} with the intention to identify flow-leaks (for security applications) and MITM attacks (for cryptography protocols) for a given sets of global variables and labels.

\begin{table}[htb]
\centering
\small
\caption{List of application developed using \textit{Pifthon}}
\label{table:tried_programs}
\begin{tabular}{c|l}
\hline
\textbf{Application Type}  & \multicolumn{1}{c}{\textbf{Applications}}\\
\hline
\multirow{5}{2cm}{Security Applications} & Meeting scheduling system \cite{laminar}\\
 & Conf. review system $\lambda\text{Chair}$ \cite{stefan2011flexible}\\
 & Password update program \cite{myers2011attacker}\\
 & WebTax example \cite{myers2000protecting}\\
 & Vickrey auction \cite{chong2006decentralized}\\
\hline
\multirow{6}{2cm}{Cryptography Protocols} &  Needham-Schroeder Public Key Protocol \cite{needham1978using}\\
 & Neuman-Stubblebine Protocol \cite{ord1993note}\\
 & Yahalom Protocol \cite{abadi1990logic}\\
 & Wide-Mouth-Frog Protocol \cite{burrows1989logic}\\
 & Otway-Rees Protocol \cite{otway1987efficient}\\
 & Diffie-Hellman Protocol \cite{diffie1976new}\\
\hline
\end{tabular}
\end{table}

{\bf Usability of \textit{Pifthon}:}
Following features of \textit{Pifthon} improve usability and meet the demands of researchers \cite{jifclipse,hicks2006languages,johnson2015exploring,hammer2010experiences}:\\
$\bullet$ $pc$-monotonic approach enables tracking $pc$-label at any program point, including nested selection and iteration statements;\\
$\bullet$ It has provision to immutably label I/O channels catering to the design of distrusted communication channels;\\
$\bullet$ Since the labelling mechanism is compositional and preserves end-to-end non-interference -- allowing data sharing between mutually distrusted distributed applications;\\
$\bullet$ A single input file contains the label specification for an application that eases alteration in security policies in one place and helps propagate changes throughout the program.

\section{Comparison with Related Work}
\label{sec:related_work}
Here, we present a comparative study in terms of labelling approaches with prominent IFC tools and flow-security against different classes of information channels such as termination-channels.

\noindent{\textbf{Information Flow Control Tools:}\\}
Language-based security literature witnesses a plethora of IFC tools and platforms that have been developed in the last decade to enforce confidentiality and integrity in prevailing programming languages based on the different flow security models \cite{denning1975secure,myers2000protecting,stefan2011disjunction}. For example, security-typed languages and monitors such as Jif \cite{myers2001jif}, JOANA \cite{tooljoana2013atps}, Paragon \cite{broberg2013paragon} for Java, FlowFox \cite{de2012flowfox}, JSFlow \cite{hedin2014jsflow}, IFC4BC \cite{bichhawat2014information} for JavaScript, FlowCaml \cite{simonet2003flow} for Caml, $\lambda_{DSec}$ \cite{zheng2007dynamic} for lambda calculus, LIO \cite{stefan2011flexible}, HLIO \cite{buiras2015hlio}, Flamio \cite{pedersen2019programming} for Haskell and SPARK flow analysis \cite{barnes2003high} for SPARK, and flow secure platforms for instance, Jif/split \cite{zdancewic2002secure}, Asbestos \cite{efstathopoulos2005labels}, HiStar \cite{histar}, Flume \cite{flume}, Aeolus \cite{cheng2012abstractions} are some of the prominent ones. Table \ref{tools_comparison} shows a comparison among some of the secure languages and monitors in terms of different label binding mechanisms (cf. Section \ref{sec:introduction}).

\begin{table}[!ht]
\small
\begin{center}
\caption{Comparison of IFC tools and platforms}
\label{tools_comparison}
\begin{tabular}{c|c|c|c}
\hline
Tools & Labelling & Flow & Termination \\ 
and Platforms & mechanism & -sensitive & -sensitive\\
\hline
Jif & $P_1$ & \xmark & \xmark \\
\hline
Paragon & $P_1$ & \xmark & \xmark \\
\hline
FlowCaml & $P_1$ & \xmark & \xmark \\
\hline
$\lambda_{DSec}$ & $P_3$ & \cmark & \xmark \\
\hline
\texttt{LIO} & $P_2$ & \xmark & \xmark \\
\hline
$\lambda_{l,FS}^{\texttt{LIO}}$ & $P_2+P_4$ & \cmark & \xmark \\
\hline
Aeolus & $P_1$ & \xmark & \xmark \\
\hline
Flamio & $P_2$ & \xmark & \xmark \\
\hline
Pifthon & $P_2+P_4$ & \cmark & \cmark \\
\hline
\end{tabular}
\end{center}
\end{table}

\noindent{\textbf{Comparison with IFC tools:\\}}
Jif \cite{myers2001jif} is one of the prominent security typed-languages that follow a static flow analysis with $pc$-reset approach for a subset of Java. Jif follows labelling mechanism $P_1$ that often result in false-positives -- a secure program shown in Table \ref{table:dynamic_labelling} may be identified as insecure by Jif. Further, the semantics of Jif cannot identify flow-leaks due to termination-and progress-channels.

JOANA \cite{tooljoana2013atps} performs a compelling flow analysis for Java, based on a system dependence graph \cite{hammer2009flow} that overcomes some of the shortcomings of Jif. In a dependence graph, each statement is a node, and an edge represents a dependency relation between two statements. Intuitively, a path in the graph that originates at a \textit{high} node and ends at \textit{low} indicates a violation of flow security. However, the flow analysis does not create any dependency between a loop and the following node unless they share an object. This limitation will not detect information leaks due to termination-channels, that makes it termination-insensitive.  

Zheng \textit{et al.} \cite{zheng2007dynamic} were the first to introduce a dynamic flow-sensitive analysis in their work on a security-typed $\lambda$-calculus, $\lambda_{DSec}$, that supports first-class dynamic labels that can be checked and changed at run-time. Contrary to \textit{Pifthon}, $\lambda_{DSec}$ retains a static $pc$-label that acts as an upper bound of the caller. Thus, the language falls in class $P_3$. Moreover, the non-interference property discussed in $\lambda_{DSec}$ is termination-insensitive.

A labelled IO Haskell library \texttt{LIO} \cite{stefan2011flexible} shares a common paradigm but subsumes the results of $\lambda_{DSec}$. Analogous to program-counter, \texttt{LIO} maintains a \textit{current label} that is mutable and evaluated as the upper bound of all the objects and inputs observed at run-time. While performing an output operation, \texttt{LIO} ensures the current label \textit{can-flow-to} target object label, thus enforces flow security.  

Although our work overlaps a bit with that of \texttt{LIO}, subtle differences are as follows: (i) compared to compile-time analysis, \texttt{LIO} follows a run-time \textit{floating-label} approach, (ii) \texttt{LIO} is termination-insensitive \cite{stefan2012addressing}, and (iii) unlike achieving flow-senstivity for each intermediate object, \texttt{LIO} is flow-sensitive for the current label only, hence follows class $P_2$. 

To overcome flow-insensitivity in \texttt{LIO}, an extension introduces a  meta-label for each reference label \cite{buiras2014dynamic}. Meta-label describes the confidentiality of  \texttt{LIO} reference label. Then upgrading a reference label also considers the meta-labels besides the label on the data stored in the reference. Authors presented a precise semantics for an extension of $\lambda$-calculus called $\lambda_{l,FS}^{\texttt{LIO}}$. However, contrary to \textit{Pifthon}, in the sequential context, the non-interference followed by this flow-sensitive extension of \texttt{LIO} is termination-insensitive.

Another Haskell library, i.e., Flamio \cite{pedersen2019programming}, enforces coarse-grained dynamic information flow control, called FLAM \cite{arden2015flow}. Flamio is a language instantiation of FLAM into LIO that can leverage it's decentralized authorization model and distributed proof search of trust relationships. However, similar to \texttt{LIO}, Flamio is flow-sensitive in \textit{computation context} only, hence follows class $P_2$, and non-interference is termination-insensitive.

Complementing recent IFC efforts that follow a run-time analysis, e.g., \texttt{LIO} \cite{stefan2011flexible}, $\lambda_{l,FS}^{\texttt{LIO}}$ \cite{buiras2014dynamic}, \textit{Pifthon} performs a compile-time flow analysis for a dialect of Python language. The run-time analysis has the following advantages: (i) it allows running non-interfering executions selectively, providing an essence of ``lazy'' analysis; (ii) it enables users to perceive control dependency precisely. For instance, the analysis could accurately identify the source index of a list data structure or the source information that depends on the user inputs. Whereas, the compile-time approaches, including \textit{Pifthon}, aim to capture all possible flow channels, including the non-executing control paths, thus mitigate the risk of information leaks. However, it may not accurately identify the information source that depends on the expression resolution or user inputs.

From the above comparative study, we can infer that  \textit{Pifthon} stands out in enabling the user to precisely determine the leaks, in the context of flow-sensitive information flow analysis that includes termination- and progress-sensitive flow channels.
\section{Conclusions}
\label{conclusion}
In this paper, we have proposed an information flow analysis approach based on a semi-dynamic (hybrid) labelling. Using the approach, we have built {\it Pifthon\/} for analyzing information flow violations for a variant of Python language called PyX. The platform can be easily  adapted  for a variety of imperative programming languages. The security platform is quite precise as compared to various other platforms and is capable of analyzing termination-sensitive programs. We have demonstrated applications of the platform for a variety of typical security-critical programs and cryptography protocols used for security analysis. We have been using the platform for analyzing programs of reasonable size and found quite encouraging as it finds flow violations that help the user to fix the same. Our approach is proved to be sound in terms of classical non-interference freedom. 

Further work is in progress to include features like exception handling, operations on mutable data structures, object creation, and initialization, etc., that invariably demand run-time analysis in \textit{Pifthon} and to extend the platform for concurrent programs capable of performing IO.
\bibliographystyle{ACM-Reference-Format}
\bibliography{main}

\newpage
\appendix
\section{Termination of $\mathbf{Pifthon}$ Labelling Mechanism}
\label{append:proposition}
\begin{prop}
For any given PyX program $c$ not containing iteration, and execution state $[\sigma,\lambda]$, the evaluation of $\Gamma\vdash\langle c,\sigma,\lambda\rangle$ in \textit{Pifthon} always terminates. 
\label{prop:wo_loop_termination}
\end{prop}
\begin{proof}
The proof can be shown by structural induction. The proposition holds good for the base cases \texttt{pass}, and $x=e$. For the inductive step, it is easy to prove that if the proposition holds for $c_1$ and $c_2$, then it also holds for conditional and sequence statements.
\end{proof}

\begin{prop}
For any given statement $c$ that does not contain nested loop, and execution context $[\sigma,\lambda]$, the evaluation $\Gamma\vdash\langle\texttt{while }e:c,\sigma,\lambda\rangle$ in \textit{Pifthon} always terminates.
\label{prop:loop_termination}
\end{prop}
\begin{proof}
From the Pifthon labelling semantics, we note that the only case in which the evaluation of $\Gamma\vdash\langle\texttt{while }e:c,\sigma,\lambda\rangle$ does not terminate is when either the evaluation of $\Gamma\vdash\langle c,\sigma,\lambda'\rangle$ does not terminate, or $\Gamma\vdash\langle\texttt{while }e:c,\sigma,\bullet\rangle$ goes into an infinite recursion.

\noindent The former is not possible due to Proposition \ref{prop:wo_loop_termination}. Impossibility of the latter is shown by considering the evaluation of $\Gamma\vdash\langle\texttt{while }e:c,\sigma,\lambda\rangle$:
\begin{enumerate}
  \item $\Gamma\vdash\langle e,\sigma,\lambda\rangle\downarrow v^l$, $\lambda_1(pc) = l\oplus\lambda(pc)$.
  \item $\Gamma\vdash\langle c,\sigma,\lambda_1\rangle\rightarrow\langle\epsilon,\sigma_1,\lambda_2\rangle$.
  \item If $\lambda_2 = \lambda_1$ the evaluation terminates and there is nothing to prove. So we assume that $\lambda_2 \neq \lambda_1$. In this case $\Gamma\vdash\langle\texttt{while }e:c,\sigma,\lambda_2\rangle$ is invoked which proceeds as follows.
  \item $\Gamma\vdash\langle e,\sigma_1,\lambda_2\rangle\downarrow v^l$, $\lambda_3(pc) = l\oplus\lambda_2(pc)$.
  \item $\Gamma\vdash\langle c,\sigma_1,\lambda_3\rangle\rightarrow\langle\epsilon,\sigma_1,\lambda_4\rangle$.
  \item If $\lambda_4 = \lambda_3$ the evaluation terminates and there is nothing to prove. So we assume that $\lambda_4 \neq \lambda_3$. In this case $\Gamma\vdash\langle\texttt{while }e:c,\sigma,\lambda_4\rangle$ is invoked which proceeds as follows.
  \item $\Gamma\vdash\langle e,\sigma_1,\lambda_4\rangle\downarrow v^l$, $\lambda_5(pc) = l\oplus\lambda_4(pc)$. 
  \item $\Gamma\vdash\langle c,\sigma_1,\lambda_5\rangle\rightarrow\langle\epsilon,\sigma_1,\lambda_6\rangle$.
\end{enumerate}

We claim that $\lambda_6 = \lambda_5$. The proof is given below.

\begin{enumerate}
\item First iteration:
		\begin{equation*}
		\begin{split}
		    l &= \lambda(pc)\oplus_{v\in var(e)\cap G}\lambda(v)\\
		    \lambda_1(pc) &= l\oplus\lambda(pc)\\
			\lambda_2(pc) &= \lambda_1(pc) \oplus_{v\in SV(c)\cap G}\lambda(v)\\
 				&= \lambda(pc) \oplus_{v\in var(e)\cap G}\lambda(v) \oplus_{v\in SV(c)\cap G}\lambda(v)\\
 				&= \lambda(pc) \oplus_{v\in (var(e)\cup SV(c))\cap G}\lambda(v)\\
		\end{split}
		\end{equation*}
\item Second iteration:
			\begin{equation*}
			\begin{split}
			\lambda_3(pc) &= \lambda_2(pc)\oplus_{v\in var(e)\cap G}\lambda(v)\\
				&= \lambda(pc)\oplus_{v\in (var(e)\cup SV(c))\cap G}\lambda(v)\\
				& \oplus_{v\in var(e)\cap G}\lambda(v)\\
				&= \lambda(pc)\oplus_{v\in (var(e)\cup SV(c))\cap G}\lambda(v)\\
			\end{split}
			\end{equation*}
			\begin{equation*}
			\begin{split}
			\lambda_4(pc) &= \lambda_3(pc)\oplus_{v\in SV(c)\cap G}\lambda(v)\\
				&= \lambda(pc)\oplus_{v\in (var(e)\cup SV(c))\cap G}\lambda(v)\\
			\lambda_4(x) &= \lambda_3(x)\oplus\lambda_3(pc)\oplus_{v\in SV(c)\cap G}\lambda(v)\\
				&= \lambda_3(x)\oplus\lambda_3(pc)\\
				& \text{This is because the label of $pc$ is already }\\
				& \text{influenced by all the global variables in } c.\\
				&= \lambda_2(x)\oplus\lambda(pc)\oplus_{v\in (var(e)\cup SV(c))\cap G}\lambda(v)
			\end{split}
			\end{equation*}
\item Third iteration:
			\begin{equation*}
			\begin{split}
			\lambda_5(pc) &= \lambda_4(pc)\oplus_{v\in var(e)\cap G}\lambda(v)\\
				&= \lambda(pc)\oplus_{v\in (var(e)\cup SV(c))\cap G}\oplus_{v\in var(e)\cap G}\lambda(v)\\
				&= \lambda(pc)\oplus_{v\in (var(e)\cup SV(c))\cap G}\lambda(v)\\
			\lambda_5(x) &= \lambda_4(x)\\
				&= \lambda_2(x)\oplus\lambda(pc)\oplus_{v\in (var(e)\cup SV(c))\cap G}\lambda(v)
			\end{split}
			\end{equation*}
			\begin{equation*}
			\begin{split}
			\lambda_6(pc) &= \lambda_5(pc)\oplus_{v\in SV(c)\cap G}\lambda(v)\\
				&= \lambda(pc)\oplus_{v\in (var(e)\cup SV(c))\cap G}\lambda(v)\\
			\lambda_6(x) &= \lambda_5(x)\oplus\lambda_5(pc)\oplus_{v\in SV(c)\cap G}\lambda(v)\\
				&= \lambda_2(x)\oplus\lambda(pc)\oplus_{v\in (var(e)\cup SV(c))\cap G}\lambda(v)
			\end{split}
			\end{equation*}
\end{enumerate}
It can be observed that $\lambda_6 = \lambda_5$. Thus,  we can conclude that for the iteration statement, \textit{Pifthon} labelling procedure terminates after a maximum of three iterations.
\end{proof}

\section{Labelling the Password Update Program}
\label{sec:pwd_update_labelling}
\textit{Pifthon} performs the following operations as shown in Table \ref{tbl:pwdupdate}. At line 3, $pc$ label become equal to LUB of the labels of $pwd\_db$ and $guess\_pwd$. Next it checks the constraint for allowing the implicit flow from if-guard to the target variable $pwd\_db$. At line 4, as target $pwd\_db$ is a global variable, therefore, \textit{Pifthon} checks if the condition $\lambda(pc)\oplus\lambda(new\_pwd)\oplus cl\leqslant\lambda(pwd\_db)$ holds true. Finally, at line 5, label of local variable $result$ is updated by the $pc$ label, thus obtain the final label of $result$.

\begin{table}[H]
\caption{Steps performed by \textit{Pifthon} for computing final label of $result$}
\label{tbl:pwdupdate}
\begin{tabular}{cl}
\hline
{\bf Line No.} & \multicolumn{1}{c}{\textbf{Label Computations \& Constraints}} \\
\hline
2 & Initial label of $pc$ is $\lambda(pc)=$ (B,\{*\},\{ \})\\
  & \\
3 & (i) $l=\lambda(pwd\_db)\oplus\lambda(guess\_pwd)=$(B,\{B\},\{A,B\})\\
  & (ii) $\lambda(pc)=l\oplus\lambda(pc)=$(B,\{B\},\{A,B\})\\
  & (iii) $\lambda(pc)\leqslant\lambda(pwd\_db)$\\
  & \\
4 & (i) $\lambda(pc)=\lambda(new\_pwd)\oplus\lambda(pc)=$ (B,\{B\},\{A,B\});\\
  & (ii) $\lambda(pc)\oplus cl\leqslant\lambda(pwd\_db)$\\
  & \\
5 & $\lambda(result)=\lambda(pc)=$(B,\{B\},\{A,B\})\\  
\hline
\end{tabular}
\end{table}

\section{Output of $\mathbf{Pifthon}$ For Programs}
\label{sec:pifthon_output}
\begin{figure}[H]
    \centering
    \includegraphics[width=8.5cm, height=4cm, keepaspectratio]{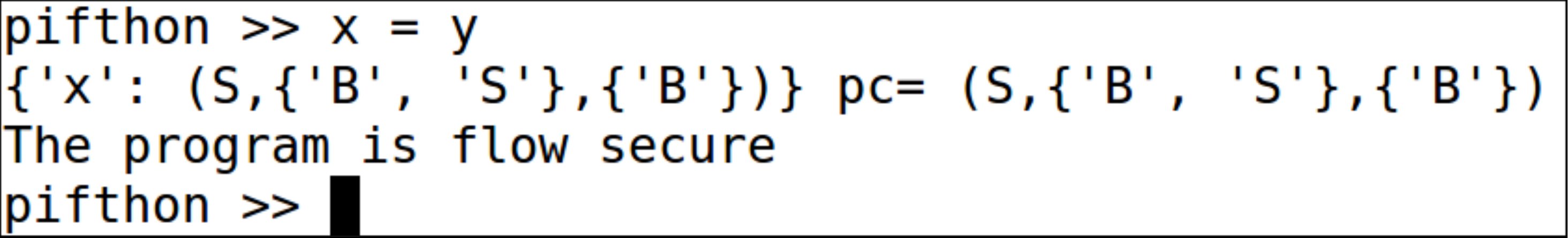}
    \caption{Output of an assignment $x=y$, where $x$ is a local variable, $\lambda(y)=(B,\{B,S\},\{B\})$, and $cl=(S,\{S\},\{A,B\})$ (Cf. Table \ref{tbl:assignment})}
    \label{fig:assignment_safe}
\end{figure}

\begin{figure}[H]
    \centering
    \includegraphics[width=8.5cm, height=4cm, keepaspectratio]{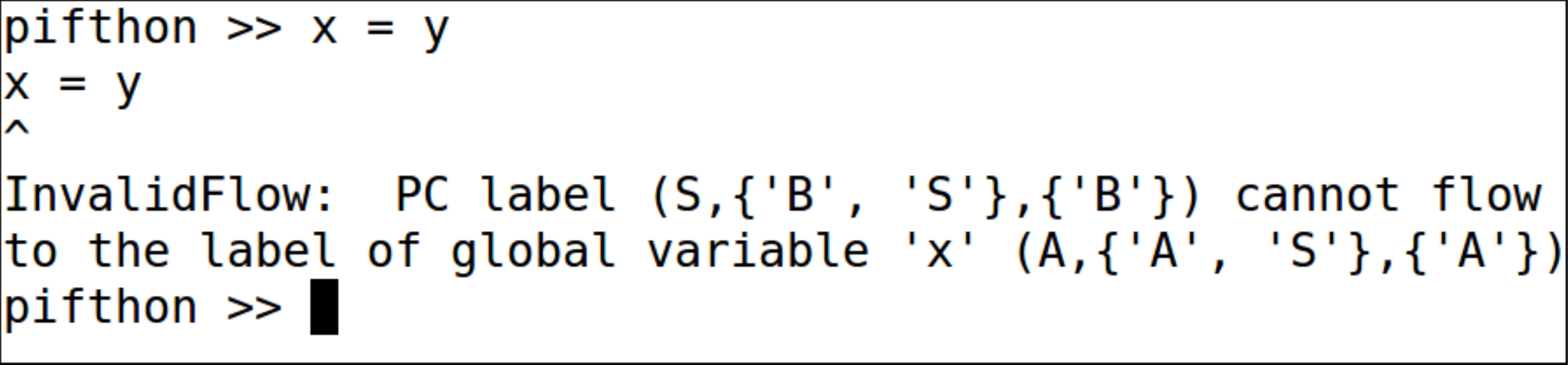}
    \caption{Output of an assignment $x=y$, where $x$ is a global variable, $\lambda(x)=(A,\{A,S\},\{A\})$ $\lambda(y)=(B,\{B,S\},\{B\})$, and $cl=(S,\{S\},\{A,B\})$ (Cf. Table \ref{tbl:assignment})}
    \label{fig:assignment_error}
\end{figure}

\begin{figure}[H]
    \centering
    \includegraphics[width=8.5cm, height=4cm, keepaspectratio]{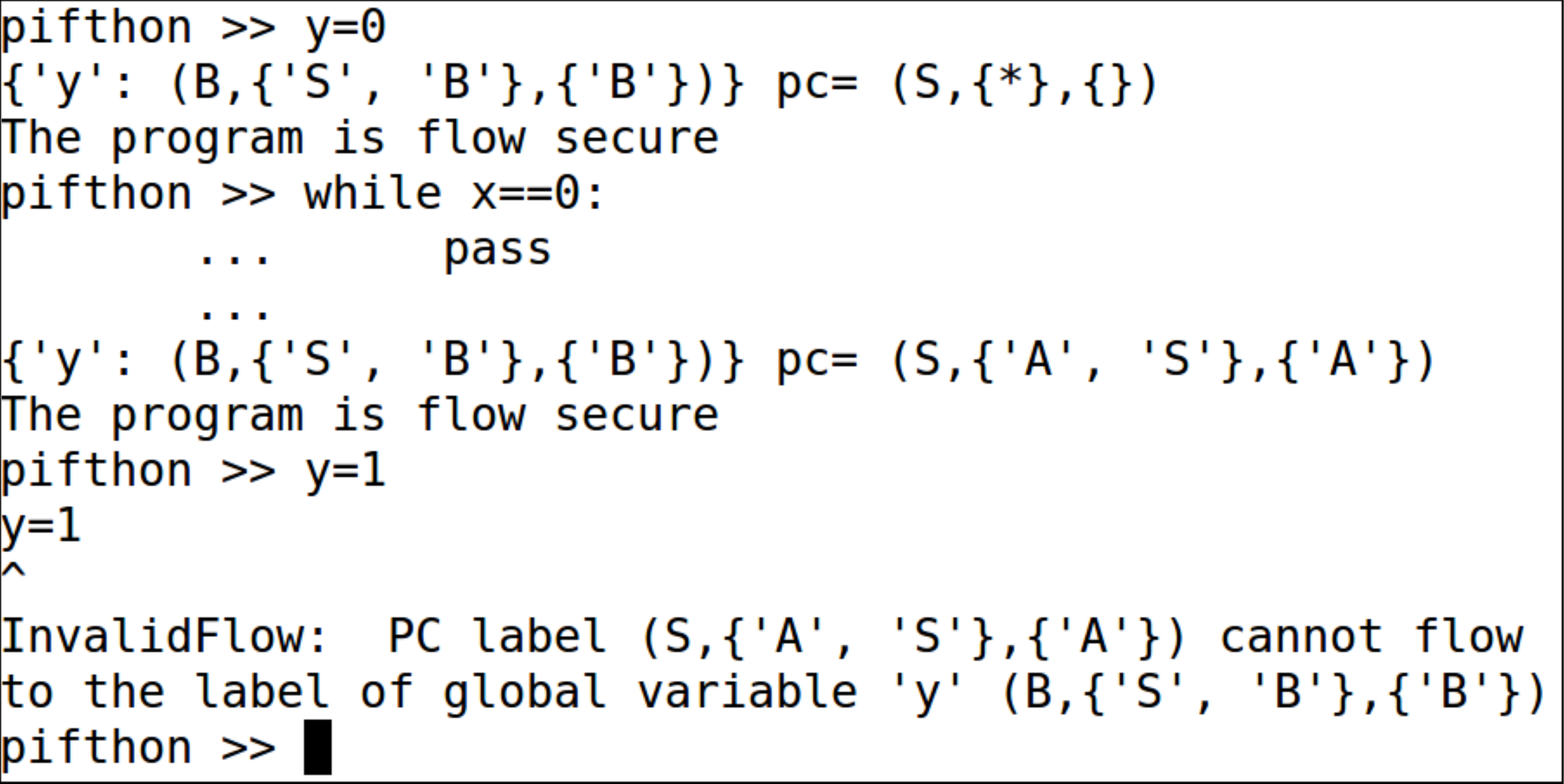}
    \caption{\textit{Pifthon} throws error for a termination-sensitive information leak (Cf. Figure \ref{table:example_iteration})}
    \label{fig:non_termination_error}
\end{figure}
\end{document}